\documentclass[12pt]{article}
 \usepackage{amsmath, amsthm, amssymb, bbm, setspace,bigints}
 \usepackage[margin=1 in]{geometry}
\usepackage{caption, enumitem}
\usepackage{subcaption}
\usepackage[toc,page]{appendix}
\usepackage{graphicx, amsfonts, tikz, multirow}
\usepackage{natbib}
\usetikzlibrary{bayesnet}
\usepackage{pdfpages}
\usepackage{epstopdf}
\usepackage{booktabs}
\usepackage[ruled]{algorithm2e}
\usepackage[colorlinks,citecolor=blue,urlcolor=blue]{hyperref}
\usepackage[utf8]{inputenc}
\usepackage{authblk}
\usepackage{float}

\def\T{\mathrm{\scriptscriptstyle{T}}}

\def\|{\mid}

\newtheorem{theorem}{Theorem}[section]

\newtheorem{proposition}[theorem]{Proposition}

\def \T {\text{T}}
\graphicspath{{plots/}}

\date{}
\title{Inferring synergistic and antagonistic interactions in mixtures of exposures}

\author[1]{Shounak Chattopadhyay\thanks{shounak.chattopadhyay@ucla.edu}}
\author[2]{Stephanie M. Engel\thanks{stephanie.engel@unc.edu}}
\author[3]{David Dunson \thanks{dunson@duke.edu}}
\affil[1]{Department of Biostatistics, University of California, Los Angeles}
\affil[2]{Department of Epidemiology, Gillings School of Global Public Health, University of North Carolina at Chapel Hill}
\affil[3]{Departments of Statistical Science and Mathematics, Duke University}

\begin{document}
\maketitle
\begin{abstract}
There is abundant interest in assessing the joint effects of multiple exposures on human health.  This is often referred to as the mixtures problem in environmental epidemiology and toxicology.  Classically, studies have examined the adverse health effects of different chemicals one at a time, but there is concern that certain chemicals may act together to amplify each other's effects.  Such amplification is referred to as {\em synergistic} interaction, while chemicals that inhibit each other's effects have {\em antagonistic} interactions.  Current approaches for assessing the health effects of chemical mixtures do not explicitly consider synergy or antagonism in the modeling, instead focusing on either parametric or unconstrained nonparametric dose response surface modeling.  The parametric case can be too inflexible, while nonparametric methods face a curse of dimensionality that leads to overly wiggly and uninterpretable surface estimates. We propose a Bayesian approach that decomposes the response surface into additive main effects and pairwise interaction effects, and then detects synergistic and antagonistic interactions. Variable selection decisions for each interaction component are also provided.
This Synergistic Antagonistic Interaction Detection (SAID) framework is evaluated relative to existing approaches using simulation experiments and an application to data from NHANES.
\end{abstract}


\section{Introduction}\label{sec:introduction}

There is considerable concern that humans are exposed to a variety of potentially adverse chemicals, and these exposures may have adverse health effects.  Classically, such health effects have been assessed through one exposure at a time studies, either  collecting {\em in vitro} or {\em in vivo} data at different doses of a single chemical or  focusing analyses of observational epidemiology studies on single exposures. Then, in order to predict the overall health effect of a mixture of different exposures, one needs to make strong assumptions, such as additivity and independence of exposures.  Unfortunately, such predictions will misestimate an individual's true adverse health risk if certain chemicals interact.  Of particular concern are {\em synergistic interactions} in which the adverse effect of one chemical is increased due to the presence of another chemical.  If regulatory agencies are 
unaware of such synergistic interactions in setting pollution guidelines, they may inadvertently permit substantial pollution-induced mortality and morbidity.  On the other hand, {\em antagonistic interactions} may lead to additive models over-predicting risk in which case certain chemicals may be over-regulated.

In this paper, our goal is to analyze the effect of heavy metal exposure on human kidney function. The impact of metal exposures on human renal function has garnered considerable attention from the epidemiological community. A review of the existing literature shows evidence of degraded kidney function following prolonged exposure to heavy metals \citep{pollack2015kidney, luo2020metal}. 
Most of this literature is based on simple one exposure-at-a-time correlation analyses from observational epidemiology data. It is clearly of interest to study joint effects of multiple metals, while adjusting for covariates that may act as potential confounding variables. Mechanistically, it seems unlikely that metals have a simple additive effect on kidney function. For example, healthy kidney function may continue with a single metal exposure at relatively low doses, but as dose increases and/or additional metal exposures are added it is likely that kidney function may worsen rapidly. Such a dose response surface would be reflective of a synergistic interaction.
In our analyses, we focus on 2015-16 data from NHANES. These data contain information on heavy metals found in spot urine collections and also on urine creatinine, which can be used as a marker of kidney function when urinary dilution issues are appropriately accounted for.

There is a need for new statistical tools for accommodating interactions in assessing health effects of mixtures of chemical exposures.  For recent developments, refer to 
 \cite{joubert2022powering}.
A broad set of strategies have been taken in this literature. The first extends generalized linear models to include a quadratic term characterizing pairwise interactions. 
\cite{herring2010nonparametric} takes a Bayesian approach to obtain posterior inclusion probabilities for each interaction. For recent more elaborate approaches, refer to  \cite{ferrari2021bayesian, wang2019penalized}.  Such methods can identify synergistic versus antagonistic interactions through signs of the quadratic coefficients, but rely on a restrictive parametric model.  Nonparametric alternatives include 
Bayesian Additive Regression Trees (BART) \citep{chipman2010bart} and Gaussian processes (GP) \citep{williams2006gaussian}. Bayesian Kernel Machine Regression (BKMR) \citep{bobb2015bayesian} implements GP regression with variable selection motivated by the mixtures problem. These approaches often produce wiggly and difficult to interpret dose response surfaces. To alleviate such issues, \cite{antonelli2020estimating} 
and \cite{samanta2022estimation}
nonparametrically model nonlinear pairwise interactions while conducting variable selection, with the latter article controlling false discovery rates in interaction detection. However, none of these nonparametric approaches identify synergistic or antagonistic interactions.
A third strategy is focused on bridging between parametric and nonparametric approaches. As an example, MixSelect \citep{ferrari2020identifying}
includes a nonparametric deviation from a quadratic regression using a Gaussian process, constrained to be orthogonal to the quadratic regression.  To simplify dose response modeling, 
 \cite{molitor2010bayesian} propose  profile regression based on clustering exposures, while \cite{czarnota2015assessment} use a single index model based on a weighted sum of the exposures. Neither approach allows for inferences on interactions.

We propose a nonparametric Bayesian approach for identifying synergistic or antagonistic interactions between $p$ exposures. The dose response surface is decomposed additively into $p$ main effects and $\binom{p}{2}$ pairwise interactions. 
Ruling out higher order interactions substantially reduces dimensionality, while aiding interpretability.  
Each pairwise interaction is decomposed into the difference of two non-negative functions, facilitating selection of synergistic or antagonistic effects.
The proposed {\em synergistic antagonistic interaction detection} (SAID) approach is shown to improve over unconstrained nonparametric regression, leading to superior performance in simulation studies. To avoid modeling complicated pairwise surfaces, we focus on a special case to dramatically reduce dimensionality and obtain substantial computational gains. The resulting interaction surface retains a high degree of flexibility in capturing nonlinear surfaces. We also outline an approach to carry out variable selection on the pairwise interactions. 
The SAID approach employs a computationally efficient Markov chain Monte Carlo (MCMC) algorithm to sample from the posterior distribution.

The paper is outlined as follows. We describe the urine creatinine and metal exposure data in Section \ref{sec:application} and outline challenges involved with analyzing these data. Section \ref{sec:methodology} provides details of our SAID approach. 
Section \ref{sec:simulation} compares SAID and existing competitors in simulation experiments, highlighting benefits of using SAID in terms of estimation accuracy, valid uncertainty quantification, and variable selection. Section \ref{sec:SIMapplication}
uses SAID to investigate synergistic and antagonistic interactions of metal exposures in predicting urine creatinine levels based on NHANES 2015-16 data. Section \ref{sec:discussion} includes concluding remarks.

\section{Kidney Function Data Analysis}\label{sec:application}

\subsection{Motivation}
\label{sec:motivation}

We are interested in assessing the impact of exposure to heavy metals on human renal function. The concentration of creatinine in the blood or urine is an established biomarker of kidney function \citep{kashani2020creatinine, barr2005urinary}. Exposure to heavy metals has been linked to  changes in renal function and kidney damage \citep{pollack2015kidney}.  In addition to being indicative of renal function, creatinine levels are also related to muscle mass  \citep{forbes1976urinary, baxmann2008influence}. Urine creatinine has been shown to be positively associated with serum creatinine levels
in studies excluding individuals with 
chronic kidney disease (CKD) 
 \citep{jain2016associated}. Existing studies, such as the one in \cite{luo2020metal}, show a statistical association between biomarkers of chronic kidney disease and blood levels of the heavy metals cobalt (Co), chromium (Cr), mercury (Hg), and lead (Pb) using NHANES data. \cite{kim2015environmental} find an association between CKD and elevated levels of cadmium (Cd) in blood. These studies primarily focus on marginal associations between individual chemicals and CKD. Our focus is instead on studying how multiple heavy metal exposures relate to kidney function measured through urine creatinine, with an emphasis on identifying synergistic and antagonistic interactions.

\subsection{Data Description}
\label{sec:data-description}

We analyze data that were collected by NHANES for the year $2015$. We consider urine analyte levels of the $13$ heavy metals: Antimony (Sb), Barium (Ba),  Cadmium (Cd), Cesium (Cs), Cobalt (Co), Lead (Pb), Manganese (Mn),  Molybdenum (Mo), Strontium (Sr), Thallium (Tl), Tin (Sn), Tungsten (W), and Uranium (U) as exposure variables and the level of urine creatinine (uCr) as the response variable. The unit of measurement for exposure variables is $\mu$g/mL while the unit for the response variable is mg/dL. We also adjust for age (in years), sex (male or female), ethnicity (Non-Hispanic White, Non-Hispanic Black, Mexican American, Other Hispanic, and Other), and body mass index (BMI). 

We start off by only considering the data for $2300$ individuals for which values of the response variable (uCr) are not missing. In this data set, we first remove the individuals having serious kidney disease. For this, we used the albumin-to-creatinine ratio (ACR) and removed the subjects who had albuminuria, defined as their ACR satisfying $\mbox{ACR} \geq 30$ mg/g. The resulting data set has $2008$ individuals. Furthermore, the response variable for one individual was less than the limit of detection (LOD); this singular data point was removed for purposes of analysis. For the metal exposures, there are both missing data points and data points which are lower than their respective LODs. We remove any individuals with an exposure entry missing. Finally, we removed the entries of the covariate values which were missing. In summary, neither the response variable nor the covariate variables considered in our analysis have any missing data or data less than the LOD; however, the metal exposure data contain points less than the LOD. After removing the individuals missing response, covariate, and exposure measurements, the final sample size is $n = 1979$.

We look at various summary measures of the response and exposure variables. The uCr levels are right skewed, with mean $119.67$ mg/dL, median $104.00 $ mg/dL, and semi-interquartile range (SIQR) $52.50$ mg/dL. The heavy metal exposure variables are mild to moderately correlated, with a heat plot of the  correlation matrix of the raw exposures shown in Figure \ref{fig:corplot}. More specifically, $58$ out of the $78$ pairwise correlation coefficients are less than $0.20$. The median correlation is $0.10$ with an SIQR of $0.08$, while the maximum correlation is $0.51$, between Barium and Strontium. Furthermore, there is wide variation between the range of different exposure variables. For example, the maximum recorded exposure to Cesium in the sample is $110.87$ $\mu$g/mL, while that of Uranium is only $0.77$ $\mu$g/mL. 

\begin{figure}
    \centering
    \includegraphics[scale=0.3]{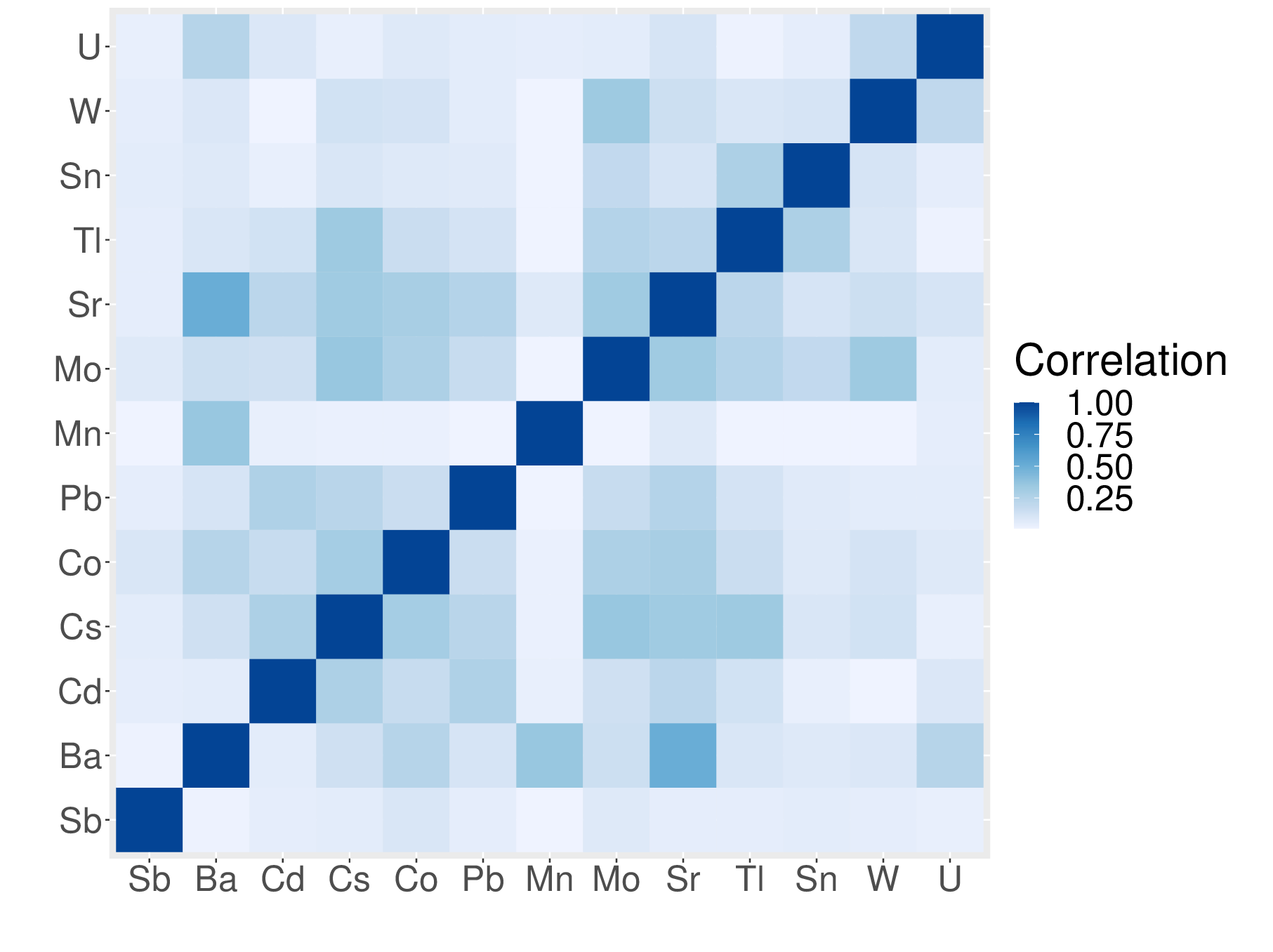}
    \caption{Heat plot showing correlations between heavy metals in NHANES 2015-16 data.}
    \label{fig:corplot}
\end{figure}


\subsection{Urinary Dilution}
\label{sec:urine-dilution}

In NHANES, the urine data is collected from spot collections. The urine sample collection procedure is regulated to ensure consistent specimen collection across differing water loadings of the subjects. Regardless, due to the nature of data collection, the concentrations of the heavy metals and creatinine levels can vary substantially across the sampled participants due to different hydration levels. For this reason, we adjust the raw response (creatinine) and exposure variables (heavy metals) measured in the urine using the urine flow rate (UFR, measured in mL/min) of the subjects. Adjusting for dilution using urine flow rate also helps reduce wide variation in exposure measurements. 
Adjusting urinary measures for variation in water loading across individuals using urine flow rates is prevalent in existing literature \citep{jeng2021clinical, middleton2016assessing, hays2015variation}. Suppose for the $i$-th subject, the urine creatinine level is $C_i$ and the urinary concentrations of the heavy metals is the vector ${\bf M}_i = (M_{i1}, \ldots, M_{ip})^{\T}$. Provided the urine flow rate for the $i$-th subject is $\tau_i > 0$, we adjust for urinary dilution by multiplying both the original response and exposure variables by $\tau_i$. That is, the urine flow adjusted response variable and exposure vector are $\tau_i C_i$ and $\tau_i {\bf M}_i$, respectively, which we refer to as the dilution-adjusted  response and exposure variables, respectively. 

 \subsection{Issues with existing approaches}
 \label{sec:issues-current-approaches}

As two different types of state-of-the-art methods, we apply MixSelect and BKMR on the heavy metals and creatinine data.
We (natural) log transform the response and exposures prior to analysis. Both BKMR and MixSelect conduct Bayesian variable selection, providing posterior inclusion probabilities (PIPs) for each exposure.
Excluding exposures having PIPs $<0.5$, 
BKMR excludes Barium (PIP $\approx$ 0), Lead (PIP $\approx$ 0), and Tin (PIP = 0.13), while MixSelect excludes Cobalt (PIP = 0.29) and Uranium (PIP = 0.11).
With these exposures excluded, BKMR produces a $10$ dimensional dose response surface in the remaining metals. Although inferences on main effects and pairwise interactions could rely on examining univariate and bivariate cross-sections of the $10$-dimensional surface, such results are exploratory and difficult to interpret. In contrast, MixSelect provides PIPs for both the main effects and pairwise interactions. However, it is not reassuring that MixSelect excludes different exposures than BKMR. In addition, all the interaction PIPs for MixSelect are $<0.11$, indicating that none of the interactions are selected.
 
We are motivated by these preliminary results to develop an approach that is flexible enough to characterize nonlinear dose response surfaces, while allowing us to formally detect pairwise interactions that are synergistic or antagonistic.  Flexible nonparametric dose response surface methods can characterize synergistic or antagonistic pairwise interactions but they tend to be hidden within a flexible multivariate surface, as highlighted for BKMR above. 
We describe our proposed Synergistic Antagonistic Interaction Detection (SAID) approach in Section \ref{sec:methodology}. SAID will be applied to the motivating metals and creatinine data in Section \ref{sec:SIMapplication}. 


\section{Structured Interaction Modeling Approach}
\label{sec:methodology}

\subsection{Basic Modeling Structure}
\label{sec:prelim}
Denote the health outcome of interest by $y_i$, for individuals $i=1,\ldots,n$; we assume $y_i \in \mathbb{R}$ but the approach can be easily modified to allow $y_i \in \{0,1\}$ or ordered categorical or count responses. We denote the exposures as 
$\mathbf{x}_i = (x_{i1}, \ldots, x_{ip})^{\T} \in [0,1]^p$, assuming without loss of generality they each fall within $[0,1]$,
and the covariates as $\mathbf{z}_i \in \mathbb{R}^q$. The model we consider is given by
\begin{equation}
    \label{eq:our-model-base}
    y_i = H(\mathbf{x}_i) + \eta^{\T} \mathbf{z}_i + \epsilon_i,\quad 
    \epsilon_i \sim N(0,\sigma^2), 
\end{equation}
where $H(\mathbf{x}) = H(x_1, \ldots, x_p)$ is the dose response function of the exposures, and we follow common practice in adjusting for the covariates linearly. 
Without loss of generality, we assume that higher values of the health outcome $y$ represent worse health in interpreting exposure effects.
As in \cite{wei2020sparse, brezger2006generalized}, we characterize the dose response function of the exposures $H({\bf x})$ via an additive expansion into main effects and pairwise interaction terms: 
\begin{equation}
\label{eq:our-model}
    H(x_1, \ldots, x_p) = \alpha + \sum_{j=1}^{p} f_j(x_j) + \sum_{1 \leq u < v \leq p} h_{uv}(x_u, x_v),
\end{equation}
where $\alpha$ is an intercept, $f_j(x_j)$ is the main effect of the $j$th exposure for $j = 1, \ldots, p$, and $h_{uv}(x_u, x_v)$ for $1 \leq u < v \leq p$ is a pairwise interaction. 
The decomposition in \eqref{eq:our-model} allows interpretation of different components in a factorization of the dose response surface $H$ into main effects and pairwise interactions. 

The primary innovation of this paper is
our approach for inferences on the pairwise interaction component.
Section \ref{sec:interaction-effects} introduces the general structure of our model for the $h_{uv}$s. Section \ref{sec:variable-selection} describes an approach for interaction selection.  Section \ref{sec:main-effects} describes our model for the main effects. Section \ref{sec:posterior-sampling} contains details on posterior computation.

\subsection{Modeling pairwise interactions}
\label{sec:interaction-effects}
We propose a model for the pairwise interactions $h_{uv}$s. Assume that $h_{uv}$ is continuous and admits finite partial derivatives up to second order. For $\mathbf{a} = (a_1, \ldots, a_d)^{\T} \in \mathbb{R}^d$, we let $\mathbf{a}^2 = (a_1^2, \ldots, a_d^2)^{\T} \in \mathbb{R}^d$. For a non-negative function $f : [0,1]^2 \to [0,\infty)$, we say $f \equiv 0$ if $f(x_1, x_2) = 0$ for all $(x_1, x_2) \in [0,1]^2$; otherwise, if $f(x_1, x_2) > 0$ for at least one $(x_1, x_2) \in [0,1]^2$, $f \not \equiv 0$. 

In conducting inferences on interactions in environmental epidemiology, it is of primary interest to assess whether exposures work together to magnify their health effects (synergy), tend to block each other's effects (antagonism), or have effectively no interaction (null).
Hence, for exposures $u$ and $v$, we focus on the following classes of interactions. 
\begin{enumerate}
    \label{def:synergistic}
\item {\em Synergistic} if 
$h_{uv}(x_u, x_v) \ge 0$ for all $(x_u, x_v) \in [0,1]^2$ with at least one strict inequality.

\item \label{def:antagonistic}
{\em Antagonistic} if
$h_{uv}(x_u, x_v) \le 0$ for all $(x_u,x_v) \in [0,1]^2$ with at least one strict inequality.
\item \label{def:not-present}
{\em Null} if 
$h_{uv}(x_u, x_v)=0$ for all $(x_u, x_v) \in [0,1]^2$.
\end{enumerate}

Our primary interest is in classifying interactions $h_{uv}$ in terms of Definitions 1-3. Although we will not rule out other types of interaction surfaces {\em a priori}, we develop inference approaches that are targeted towards this three class hypothesis testing problem, which we denote by the Synergistic-Antagonistic-Null (SAN) class.  This is accomplished in a Bayesian manner with a carefully-structured model for $h_{uv}$ that shrinks towards the space of synergistic, antagonistic, and null interactions.

If we knew {\em a priori } that $h_{uv}$ was either  synergistic or null, we could let $h_{uv} = P_{uv}$, where $P_{uv} \geq 0$. If $P_{uv}(x_u, x_v) > 0$ for at least one $(x_u, x_v) \in [0,1]^2$ then $h_{uv}$ is synergistic; otherwise, $P_{uv} \equiv 0$ and $h_{uv}$ is a null interaction. Similarly, if we knew {\em a priori} that $h_{uv}$ was either antagonistic or null, we could let $h_{uv} = - N_{uv}$, where $N_{uv} \geq 0$. However, the sign of the interaction $h_{uv}$ is usually unknown and pairwise interactions may be only approximately synergistic or antagonistic. To allow for such complexities, we let 
\begin{equation}
\label{eq:int-decomposition}
    h_{uv}(x_u, x_v) = P_{uv}(x_u, x_v) - N_{uv}(x_u, x_v),
\end{equation}
for all $(x_u,x_v) \in [0,1]^2$, where $P_{uv}, N_{uv} : [0,1]^2 \to [0, \infty)$ are non-negative functions.  If $P_{uv} \not \equiv 0$ and $N_{uv} \equiv 0$, $h_{uv}$ is synergistic; while if $P_{uv} \equiv 0$ and $N_{uv} \not \equiv 0$, $h_{uv}$ is antagonistic. If both $P_{uv} = N_{uv} \equiv 0$, $h_{uv}$ is null. Thus, shrinking either $P_{uv}$ or $N_{uv}$ or both to zero shrinks $h_{uv}$ towards the class of SAN interactions. 
Considering the overall interaction surface $$I(\mathbf{x}) = \sum_{1 \leq u < v \leq p} h_{uv}(x_u,x_v) = \sum_{1 \leq u < v \leq p} P_{uv}(x_u, x_v) - \sum_{1 \leq u < v \leq p} N_{uv}(x_u, x_v),$$ shrinking some of the interactions $h_{uv}$  towards the SAN class is equivalent to a sparse estimation problem where we impose sparsity on $P_{uv}$ and $N_{uv}$ for $1 \leq u < v \leq p$. This motivates us to develop a framework to encourage sparsity among $(P_{uv}, N_{uv})_{1 \leq u < v \leq p}$ through well-constructed prior distributions.




Instead of modeling $P_{uv}$ and $N_{uv}$ as complicated bivariate functions, we further decompose $P_{uv}$ and $N_{uv}$ as products of non-negative univariate functions:
\begin{equation}
\label{eq:interaction-rank-1}
    \begin{aligned}
    h_{uv}(x_u, x_v) & = 
    P_{uv}(x_u, x_v) - N_{uv}(x_u, x_v)\\
    & = P_{uv, 1}(x_u) P_{uv, 2}(x_v) - N_{uv, 1}(x_u) N_{uv, 2}(x_v).
    \end{aligned}
\end{equation}
This model is much more flexible than commonly used quadratic regression, which lets $h_{uv}(x_u,x_v)=\gamma_{uv}x_ux_v$. To model $P_{uv, 1}, P_{uv, 2}$, $N_{uv, 1}, N_{uv, 2}$
as flexible one-dimensional non-negative functions, we rely on squaring B-spline expansions \citep{de1978practical} as follows:
\begin{equation}
    \begin{aligned}
      P_{uv, 1}(x_u) = \{\mathbf{s}_u(x_u)^{\T} \theta_{uv, 1}\}^2 &, \quad P_{uv, 2}(x_v) = \{\mathbf{s}_v(x_v)^{\T} \phi_{uv, 1}\}^2\\
      N_{uv, 1}(x_u) =  \{\mathbf{s}_u(x_u)^{\T} \theta_{uv, 2}\}^2 &, \quad N_{uv, 2}(x_v) = \{\mathbf{s}_v(x_v)^{\T} \phi_{uv, 2}\}^2,
    \end{aligned}
    \label{eq:sqrspline}
\end{equation}
where $\mathbf{s}_u(x_u) = (s_{u1}(x_u), \ldots, s_{um}(x_u))^{\T}$ denote B-splines for the $u$-th variable, chosen so that $s_{uj}(0) = 0$ for $j=1,\ldots,m$.  This is obtained by discarding the intercept spline; refer to Section \ref{supp:bspline} of the Supplementary Material for further details. Constraining $s_{uj}(0) = 0$ for $j=1,\ldots,m$ and $u=1,\ldots,p$ ensures that the interaction $h_{uv}$ satisfies $h_{uv}(x_u, 0) = 0$ for all $x_u \in [0,1]$ and $h_{uv}(0, x_v) = 0$ for all $x_v \in [0,1]$, thereby making $h_{uv}$ identifiable. The components $P_{uv}$ and $N_{uv}$ are not identifiable individually as they enter the likelihood only via their difference $h_{uv} = P_{uv} - N_{uv}$. However, this is not an issue since our focus is only in estimating the interaction $h_{uv}$, using sparsity on the components $P_{uv}, N_{uv}$ to shrink $h_{uv}$ towards the class of SAN interactions. For further details on identifiability of pairwise interactions, we refer the reader to Sections \ref{supp:bspline} and \ref{supp:model-id} of the Supplementary Material.

To estimate $h_{uv}$ in a Bayesian manner, we define a prior distribution $\pi(\mathbf{\Psi}_{uv})$ on the basis coefficients $\mathbf{\Psi}_{uv} = (\theta_{uv,1}^{\T}, \phi_{uv,1}^{\T}, \theta_{uv,2}^{\T}, \phi_{uv,2}^{\T})^{\T}$. Our choice of $\pi(\mathbf{\Psi}_{uv})$ is motivated by encouraging sparsity amongst $P_{uv}$ and $N_{uv}$ as discussed earlier. Instead of using spike-and-slab priors \citep{george1993variable, mitchell1988bayesian, ishwaran2005spike} which lead to computational inefficiency, we consider continuous shrinkage priors on the spline coefficients. For an overview of continuous Bayesian shrinkage priors and advantages over spike-and-slab approaches, we refer the reader to \cite{bhattacharya2012bayesian}.
For $1 \leq u < v \leq p$, we first let
\begin{equation}
\label{eq:rank-1-indpt-priors}
    \begin{aligned}
    \theta_{uv, 1}, \phi_{uv, 1} \mid \tau_{uv,1}, \nu & \sim N\left(0, \nu^2 \tau_{uv, 1}^2 \Sigma_0\right),\quad \\
    \theta_{uv, 2}, \phi_{uv, 2} \mid \tau_{uv,2}, \nu & \sim N\left(0, \nu^2 \tau_{uv, 2}^2 \Sigma_0\right),
    \end{aligned}
\end{equation}
where $\Sigma_0$ is a P-spline covariance matrix as in \cite{lang2004bayesian}; for more details, refer to Section \ref{supp:pspline} of the Supplementary Material. Let $p(\mathbf{\Psi}_{uv} \mid \tau_{uv,1}, \tau_{uv,2}, \nu)$ denote the conditional prior density of $\mathbf{\Psi}_{uv}$ obtained from \eqref{eq:rank-1-indpt-priors}.
The prior is completed with hyperpriors for the local and global variance parameters:
\begin{equation}
\label{eq:int-prior}
    \tau_{uv, 1}, \tau_{uv, 2} \sim C^{+}(0,1),\quad 
    \nu \sim C^{+}(0,1),
\end{equation}
where $C^{+}(0,1)$ denotes a half-Cauchy prior. 

Priors \eqref{eq:rank-1-indpt-priors}-\eqref{eq:int-prior} are chosen to have a global-local shrinkage form motivated by the horseshoe prior \citep{carvalho2009handling}. 
Small values of the global parameter $\nu$
favor most of the interactions $h_{uv}\approx 0$, while the heavy-tailed prior on local parameters $\tau_{uv,1}$ and $\tau_{uv,2}$ allow certain interactions to have $P_{uv}$ and $N_{uv}$ components arbitrarily far from zero. Furthermore, the densities of $\tau_{uv,1}^2$ and $\tau_{uv,2}^2$ have infinite spikes at $0$, allowing either $P_{uv}$ or $N_{uv}$ to be efficiently shrunk to zero. This successfully shrinks $h_{uv}$ towards the class of SAN interactions. As a result, our formulation is able to incorporate prior information as follows. Firstly, it encourages global sparsity across the $\binom{p}{2}$ interactions as $p$ increases, leading to an overall small number of non-null interactions as is expected in practice. Secondly, it is able to target potentially non-linear synergistic or antagonistic interactions. Thirdly, as we refrain from imposing strict constraints, our approach also maintains flexibility to detect interactions that do not belong to the SAN class. We describe posterior computation with this setup in Section \ref{sec:posterior-sampling}.

We also tried other approaches to model $h_{uv}$. Instead of squaring unconstrained functions, one could potentially use a Bayesian monotone spline formulation based on restricting the sign of the basis coefficients. However, we found squaring an unconstrained function to provide better estimates, particularly when the non-negative function being modeled is near $0$. We also tried using a tensor product spline model for the bivariate functions, but using products of univariate functions had superior performance. Finally, we tried an approach to further shrink the interactions towards the class of SAN interactions using an additional penalty term in the prior, but found this approach to provide little to no improvement in simulations and the motivating application while significantly complicating posterior computation. We describe these alternative approaches in more detail in Section \ref{supp:other} of the Supplementary Material.

\subsection{Variable Selection}
\label{sec:variable-selection}
To select the non-null interactions $h_{uv}$s, we first decompose
\begin{equation}
    \label{eq:canonical-decomp}
    h_{uv}(x_u, x_v) = h_{uv}^{+}(x_u, x_v) - h_{uv}^{-}(x_u, x_v),
\end{equation}
where $h_{uv}^{+}(x_u, x_v) = \max\{h_{uv}(x_u, x_v), 0\} \geq 0$ and  $h_{uv}^{-}(x_u, x_v) = \max\{-h_{uv}(x_u, x_v), 0\} \geq 0,$ for any $(x_u, x_v) \in [0,1]^2$. 
The interaction $h_{uv}$ is synergistic if and only if $h_{uv}^{+} \not \equiv 0$ and $h_{uv}^{-} \equiv 0$, antagonistic if and only if $h_{uv}^{+} \equiv 0$ and $h_{uv}^{-} \not \equiv 0$, and null if and only if $h_{uv}^{+} = h_{uv}^{-} \equiv 0$. These conditions may be rewritten in terms of their integrals since for a continuous non-negative function $P$, $P \equiv 0$ if and only if $\int P = 0$ and $P \not \equiv 0$ if and only if $\int P > 0$.

Since we use continuous shrinkage priors to estimate $h_{uv}$, there is zero posterior probability of $h_{uv}^{+}$ or $h_{uv}^{-}$ being {\em exactly} zero. In order to carry out variable selection, we first construct {\em sparsified} MCMC samples of $\int h_{uv}^{+}$ and $\int h_{uv}^{-}$ by adapting the sequential 2-means (S2M) approach outlined in \cite{li2017variable}. The S2M algorithm performs a sequential 2-means clustering of the elements of a vector with non-negative entries. Based on the output of the S2M algorithm, we set the entries of the vector assigned to the cluster with the smaller mean as exactly zero. We carry out the S2M clustering in our setting as follows, given MCMC samples of $\int h_{uv}^{+}$ and $\int h_{uv}^{-}$.
 
For a vector $v$ having zero and positive entries,  let $\{v = 0\}$ and $\{v > 0\}$ denote the set of indices $j$ such that $v_j$, the $j$ entry of $v$, exactly equals $0$ and is positive, respectively.  Suppose $S_{uv}^{+} = \int h_{uv}^{+}$ and $S_{uv}^{-} = \int h_{uv}^{-}$. Let the posterior samples of $S_{uv}^{+}$ and $S_{uv}^{-}$ be denoted by $\mathcal{S}_{uv}^{+}$ and $\mathcal{S}_{uv}^{-}$, respectively, both being $N_{MC} \times 1$ vectors where $N_{MC}$ denotes the number of posterior samples obtained. We  stack $\mathcal{S}_{uv}^{+}$ and $\mathcal{S}_{uv}^{-}$ together as $\tilde{\mathcal{S}}_{uv} = (\mathcal{S}_{uv}^{+ \T}, \mathcal{S}_{uv}^{- \T})^{\T}$ and further stack all $\tilde{\mathcal{S}}_{uv}$ for $1 \leq u < v \leq p$, calling the resulting large vector $\tilde{\mathcal{S}}.$ We then apply the S2M approach of \cite{li2017variable} on $\tilde{\mathcal{S}}$, which clusters the entries of $\tilde{\mathcal{S}}$ into two groups. We proceed to manually set all entries of $\tilde{\mathcal{S}}$ assigned to the cluster with smaller mean as $0$. This provides us with the {\em sparsified} vector $\hat{\mathcal{S}}$ based on $\tilde{\mathcal{S}}$. Let $\hat{\mathcal{S}}_{uv}^{+}$ and $\hat{\mathcal{S}}_{uv}^{-}$ denote the {\em sparsified} MCMC samples of $\mathcal{S}_{uv}^{+}$ and $\mathcal{S}_{uv}^{-}$, respectively, obtained from appropriate subsetting of the large sparsified vector $\hat{\mathcal{S}}.$ Let $C_{uv} = \{\int h_{uv}^{+} = 0\}$ and $D_{uv} = \{\int h_{uv}^{-} = 0\}.$
 
 After this post-processing step, we now compute Monte Carlo estimates of the posterior probabilities of $h_{uv}$ being null, synergistic, and antagonistic, given by $\mathbf{P}\left(C_{uv} \cap D_{uv}\right)$, $\mathbf{P}\left(C_{uv}^{c} \cap D_{uv}\right),$ and $\mathbf{P}\left(C_{uv} \cap D_{uv}^{c}\right),$ respectively, as follows:
 \begin{equation}
     \label{eq:post-prob}
     \begin{aligned}
         \hat{\mathbf{P}}\left(C_{uv} \cap D_{uv}\right) & = |\{\hat{\mathcal{S}}_{uv}^{+} = 0\} \cap \{\hat{\mathcal{S}}_{uv}^{-} = 0\}| / N_{MC},\\
     \hat{\mathbf{P}}\left(C_{uv}^{c} \cap D_{uv}\right) & = |\{\hat{\mathcal{S}}_{uv}^{+} > 0\} \cap \{\hat{\mathcal{S}}_{uv}^{-} = 0\}| / N_{MC},\\
     \hat{\mathbf{P}}\left(C_{uv} \cap D_{uv}^{c}\right) & = |\{\hat{\mathcal{S}}_{uv}^{+} = 0\} \cap \{\hat{\mathcal{S}}_{uv}^{-} > 0\}| / N_{MC},
     \end{aligned}
 \end{equation}
 where for a finite set $A$, $|A|$ denotes its cardinality. In particular, the posterior inclusion probability (PIP) of the interaction $h_{uv}$ is calculated as
\begin{equation}
    \label{eq:PIP}
    \mbox{PIP}(h_{uv}) = 1 - \hat{\mathbf{P}}(C_{uv} \cap D_{uv}). 
\end{equation}
As a rough rule of thumb, loosely based on \cite{kass1995bayes}, one can view 
interactions having $\mbox{PIP}(h_{uv})\in (0.5,0.75)$ as barely worth a mention in terms of weight of evidence against the null, 
$\mbox{PIP}(h_{uv})\in [0.75,0.95)$ as weakly to moderately suggestive, 
$\mbox{PIP}(h_{uv})\in [0.95,0.99)$ as strong evidence, and $\mbox{PIP}(h_{uv})\ge 0.99$ as very strong evidence. Our approach for variable selection performed well across simulations and the motivating application. The additional cost of computation to implement the S2M approach is negligible relative to the cost of obtaining the original MCMC samples. Furthermore, the approach is fully automatic and requires no further tuning in order to be implemented. 


\subsection{Main Effects and Other Parameters}
\label{sec:main-effects}
For reasons of identifiability, we assume the main effects $f_1, \ldots, f_p$ either satisfy $\int_{0}^{1} f_j(x_j) \, dx_j = 0$ for each $j = 1, \ldots, p$ or $f_j(0) = 0$ for $j = 1, \ldots, p$. We call the former set of conditions as \textit{integral constraints} and the latter as \textit{origin constraints}. Suppose that for exposure $j$, $\{b_{j, 1}(\cdot), \ldots, b_{j, d}(\cdot)\}$ represent  one-dimensional basis functions, such as B-splines. To enforce the identifiability conditions, we let the functions $b_{j, 1}, \ldots, b_{j, d}$ either satisfy $\int_{0}^{1} b_{j, u}(x_j) \, dx_j = 0$ for $u=1,\ldots,d$, $j = 1, \ldots, p$ if the main effects satisfy the \textit{integral constraint}, or $b_{j, u}(0) = 0$ for $u=1,\ldots,d$, $j = 1, \ldots, p$ if the main effects satisfy the \textit{origin constraint}. We provide further details in Sections \ref{supp:bspline} and \ref{supp:model-id} of the Supplementary Material. Let $\mathbf{b}_j(x_j) = (b_{j,1}(x_j), \ldots, b_{j,d}(x_j))^{\T}$. We now model $f_j$ as 
\begin{equation}
    \label{eq:main-effects}
    f_j(x_j) = \sum_{u=1}^{d} b_{j, u}(x_j) \gamma_{j, u} = \mathbf{b}_j(x_j)^{\T} \, \mathbf{\Gamma}_j,
\end{equation}
where $\mathbf{\Gamma}_j = (\gamma_{j,1}, \ldots, \gamma_{j,d})^{\T}$. To estimate the coefficient vector $\mathbf{\Gamma}_j$, we assume a univariate P-spline prior distribution \citep{lang2004bayesian} on $\mathbf{\Gamma}_j$. 
With an appropriate positive-definite matrix $\Sigma_{M}$, we can write this prior as
\begin{equation}
    \label{eq:main-effect-priors-all}
    \mathbf{\Gamma}_j \mid \lambda_j  \sim N\left(0, \frac{\Sigma_{M}}{\lambda_j}\right),\quad 
    \lambda_j  \sim G(a, a),
\end{equation}
where $\lambda_j > 0$ is a scale parameter corresponding to the $j$-th main effect. The sensitivity to the choice of hyperparameter $a$ for the distribution of $\lambda_j$ has been investigated in \cite{lang2004bayesian}; we found $a = 0.5$ to work well across simulations and applications.

To complete prior specification for the parameters in \eqref{eq:our-model-base}-\eqref{eq:our-model}, we put vague prior distributions on the intercept $\alpha$ and covariate effects $\eta$, and a non-informative prior 
on the measurement error variance $\sigma^2$:
\begin{equation}
\begin{aligned}
\label{eq:other-priors}
  \alpha & \sim N(0, 10^{4}), \quad
    \mathbf{\eta} \sim N(0, 10^{4} \, \mathbb{I}_q),\quad
    \sigma^2 \sim \sigma^{-2}.
\end{aligned}
\end{equation}

\subsection{Posterior Sampling}
\label{sec:posterior-sampling}

We rely on a Hamiltonian Monte Carlo (HMC) \citep{neal2011mcmc, betancourt2015hamiltonian, hoffman2014no}-within-Gibbs algorithm, with HMC used to sample the interaction parameters $\mathbf{\Psi}_{uv}$ for $1 \leq u < v \leq p$, and other parameters updated in Gibbs steps. Although this approach is highly effective in our experiments, it is important to carefully choose the step size $e_0$ and step length $L_0$ in HMC. In practice, we found $e_0 \approx 0.01$ and $L_0 \sim 10$ to work well. We found HMC to provide faster convergence of the chain and larger effective sample sizes when compared with standard Metropolis-Hastings (MH) \citep{metropolis1953equation, hastings1970monte} algorithms relying on normal random walk proposals. We provide a brief description of HMC in Section \ref{supp:hmc} of the Supplementary Material.



Let $\mathbf{y} = (y_1, \ldots, y_n)^{\T} \in \mathbb{R}^n$ and $\mathbf{Z} = [\mathbf{z}_1 \mid \ldots \mid \mathbf{z}_n]^{\T} \in \mathbb{R}^{n \times q}$. Following Sections \ref{sec:interaction-effects} and \ref{sec:main-effects}, let $B_1, \ldots, B_p \in \mathbb{R}^{n \times d}$ and $S_1, \ldots, S_p \in \mathbb{R}^{n  \times m}$ be such that the $i$th row of $B_j$ and $S_j$ is given by $\mathbf{b}_j^{\T}(x_{ij})$ and $\mathbf{s}_j^{\T}(x_{ij})$, respectively. Let $\mathcal{M}$ be the block diagonal matrix given by 
$\mathcal{M} = \mbox{block-diag}(10^4, 
10^{4} \, \mathbb{I}_q,
 \Sigma_M / \lambda_1
\ldots,
 \Sigma_M / \lambda_p)$.
We also let $\mathbf{\widetilde{B}} = [\mathbbm{1}_n \mid \mathbf{Z} \mid B_1 \mid \ldots \mid B_p] \in \mathbb{R}^{n \times (1 + q + pd)}$, $\mathcal{G} = (\alpha, \eta^{\T}, \mathbf{\Gamma}_1^{\T}, \ldots, \mathbf{\Gamma}_p^{\T})^{\T}$, and for $1 \leq u < v \leq p$, we let $\mathbf{h}_{uv} = (S_u \theta_{uv,1})^2 (S_v \phi_{uv,1})^2 - (S_u \theta_{uv,2})^2 (S_v \phi_{uv,2})^2$ denote the vector of the $uv$th interaction evaluated at the data points, and $\mathbf{\Theta}_{uv} = (\mathbf{\Psi}_{uv}^{\T}, \tau_{uv,1}, \tau_{uv,2})^{\T}$ be the set of parameters for the $uv$-th interaction. For a matrix $M_0$, denote its trace by $\mbox{tr}(M_0)$. Posterior sampling then proceeds as follows.
\begin{enumerate}
    \item Sample $\mathcal{G} \mid - \sim N\left(\mathcal{A}^{-1} \dfrac{\mathbf{\widetilde{B}}^{\T}\xi}{\sigma^2}, \mathcal{A}^{-1}\right)$ using \cite{rue2001fast}, where $$\mathcal{A} = \dfrac{\mathbf{\widetilde{B}}^{\T} \mathbf{\widetilde{B}}}{\sigma^2}\ + \mathcal{M}^{-1}, \quad \mathbf{\xi} = \mathbf{y} -  \sum_{1 \leq u < v \leq p} \mathbf{h}_{uv}.$$
    \item For each $j=1,\ldots,p$, sample $\lambda_j \mid - \sim G\left(0.5 + \dfrac{d}{2}, 0.5 + \dfrac{\mathbf{\Gamma}_j^{\T} \Sigma_{M}^{-1} \mathbf{\Gamma}_j}{2}\right)$.
    \item For each $u, v$ with $1 \leq u < v \leq p$, let $\mathcal{I}_{uv} = \{(i_1,i_2) : 1 \leq i_1 < i_2 \leq p,  (i_1, i_2) \neq (u,v)\}$.
    \begin{enumerate}
        \item Define $\mathbf{\Delta}_{uv} = \mathbf{y} - \left(\mathbf{\widetilde{B}} \mathcal{G} + \displaystyle \sum_{(u',v') \in \mathcal{I}_{uv}} \mathbf{h}_{u' v'} \right)$. 
        \item Given $ \mathbf{\Delta}_{uv}$ and $\nu$, use HMC to draw one sample of $\mathbf{\Theta}_{uv}$, targeting 
        \begin{align*} \Pi(\mathbf{\Theta}_{uv} \mid  \mathbf{\Delta}_{uv}, \nu) \propto & \, \pi(\mathbf{\Psi}_{uv} \mid \tau_{uv,1}, \tau_{uv,2}, \nu) \,
        \pi_{\tau}(\tau_{uv,1}) \, \pi_{\tau}(\tau_{uv,2})  \, N(\mathbf{\Delta}_{uv} \mid \mathbf{h}_{uv}, \sigma^2 \mathbb{I}_n),
        \end{align*}
        where $\pi(\mathbf{\Psi}_{uv} \mid \tau_{uv,1}, \tau_{uv,2}, \nu)$ is as in Section \ref{sec:interaction-effects} and 
        $\pi_{\tau}(x) = 1 / \{\pi (1 + x^2)\} \mathbbm{1}\{x > 0\}$ is the $C^{+}(0,1)$ density.
        \item Update the $uv$-th interaction $\mathbf{h}_{uv} = (S_u \theta_{uv,1})^2 (S_v \phi_{uv,1})^2 - (S_u \theta_{uv,2})^2 (S_v \phi_{uv,2})^2 $.
    \end{enumerate}
    \item Following \cite{makalic2015simple}, we introduce $W$ such that $\nu^2 \mid W \sim \mbox{IG}(1/2, 1/W)$ and $W \sim \mbox{IG}(1/2,1)$, which leads to $\nu \sim C^{+}(0,1)$. 
    The full conditional updates for $\nu^2$ and $W$ are given by:
    \begin{enumerate}
        \item $\nu^2 \mid - \sim \mbox{IG}\left(\dfrac{1}{2} + 2m \dbinom{p}{2}, \dfrac{1}{W} + \dfrac{1}{2}\underset{1 \leq u < v \leq p}{\displaystyle\sum} r_{uv}\right)$, 
        where $$r_{uv} = \dfrac{\theta_{uv,1}^{\T} \Sigma_0^{-1} \theta_{uv,1} + \phi_{uv,1}^{\T} \Sigma_0^{-1} \phi_{uv,1}}{ \tau_{uv,1}^2} + \dfrac{\theta_{uv,2}^{\T} \Sigma_0^{-1} \theta_{uv,2} + \phi_{uv,2}^{\T} \Sigma_0^{-1} \phi_{uv,2}}{ \tau_{uv,2}^2}.$$
        \item $W \mid - \sim \mbox{IG}(1, 1 + \nu^{-2}).$
    \end{enumerate}
    \item Sample $\sigma^2 \mid - \sim \mbox{IG}\left(\dfrac{n}{2}, \dfrac{(\mathbf{y} - \mathbf{\mu})^{\T}(\mathbf{y} - \mathbf{\mu})}{2}\right),$ where $\mu = \mathbf{\widetilde{B}} \mathcal{G} + \displaystyle\sum_{1 \leq u < v \leq p} \mathbf{h}_{uv}.$
\end{enumerate}
We repeat Steps 1-5 for a large number of iterations, discard a burn-in, and base inference on the resulting samples. 

\section{Simulation Examples}
\label{sec:simulation}
 

\subsection{Preliminaries}
\label{sec:sim-prelim}

We carry out simulation studies comparing SAID with the competitors BKMR \citep{bobb2015bayesian}, MixSelect \citep{ferrari2020identifying}, HierNet \citep{bien2013lasso}, Family \citep{haris2016convex}, PIE \citep{wang2019penalized}, and RAMP \citep{hao2018model}. We generate data according to: 
\begin{equation}
\label{eq:model-simulation}
    y_i = H({\bf x}_i) + \epsilon_i,\quad \epsilon_i \sim N(0, \sigma_0^2),
\end{equation}
where $H$ is the exposure dose response surface decomposed as in \eqref{eq:our-model}, and $\sigma_0^2$ is the true error variance. As before, we assume that the exposures satisfy $\mathbf{x}_i \in [0,1]^p$. We consider two dimensions: $p = 2$ and $p = 10$. For $p = 2$, we compare the methods in terms of their estimation performance for varying interaction signal strength and error variance. For $p = 10$, we consider accuracy in estimation and variable selection for pairwise interactions. 

HierNet, Family, PIE, and RAMP estimate the dose response surface $H$ using quadratic regression, which also provides estimates of pairwise interaction functions using bilinear surfaces of the form $h_{uv}(x_u,x_v) = \gamma_{uv} x_u x_v$ for $1 \leq u < v \leq p$. MixSelect combines quadratic regression with an additional nonlinear deviation term in the model, constrained to be orthogonal to the quadratic regression. This approach provides estimates of pairwise interactions in the spirit of quadratic regression, while improving flexibility in capturing $H$. BKMR estimates $H$ using unconstrained Gaussian processes incorporating variable selection. Thus, we do not consider BKMR as a competitor when estimating pairwise interactions, as BKMR does not provide estimates of these interactions. 
Given training points $\{(\mathbf{x}_i, y_i)\}_{i=1}^{n}$, we first estimate $H$ and the pairwise interactions $h_{uv}$ using the relevant method and denote the estimates by $\widehat{H}$ and $\widehat{h}_{uv}$, respectively. For the Bayesian methods, we use the posterior mean as the estimator. We evaluate the true surface $H$ and true interactions $h_{uv}$ at test exposure points $\{\tilde{\mathbf{x}}_1, \ldots, \tilde{\mathbf{x}}_{n_t}\}$, denoted by $\widetilde{\mathbf{H}} = (H(\tilde{\mathbf{x}}_1), \ldots, H(\tilde{\mathbf{x}}_{n_t}))^{\T}$ and $\widetilde{\mathbf{h}}_{uv} = (h_{uv}(\tilde{x}_{1u}, \tilde{x}_{1v}), \ldots, h_{uv}(\tilde{x}_{n_t, u}, \tilde{x}_{n_t, v}))^{\T}$, respectively. Finally, we estimate $\widetilde{\mathbf{H}}$ and $\widetilde{\mathbf{h}}_{uv}$ for $1 \leq u < v \leq p$ by replacing $H$ and $h_{uv}$ with $\widehat{H}$ and $\widehat{h}_{uv}$, respectively. The estimates are evaluated using root mean squared error (RMSE). 

We also consider the {\em variable selection accuracy} of competing methods. For each method, we estimate the case $1$ and case $2$ errors from classifying an interaction as synergistic, antagonistic, or null. For any category, the case $1$ error probability is given by the probability of misclassifying an interaction as not belonging to that category, when in truth the interaction is in that category. The case $2$ error probability is given by the probability of misclassifying an interaction as belonging to that category when in truth, it does not belong to that category. For SAID, we classify $h_{uv}$ to the category with the highest posterior probability calculated as in Section \ref{sec:variable-selection}, with the categories being synergistic, antagonistic, null, and none of these three. For the frequentist quadratic regression approaches HierNet, Family, PIE, and RAMP, we classify $h_{uv}$ as synergistic, antagonistic, or null according to whether $\hat{\gamma}_{uv} > 0$, $\hat{\gamma}_{uv} < 0$, or $\hat{\gamma}_{uv} = 0$, respectively, where $\hat{h}_{uv}(x_u, x_v) = \hat{\gamma}_{uv} x_u x_v$ is the estimate of $h_{uv}$ obtained from the method. For the Bayesian quadratic regression approach MixSelect, we obtain posterior probabilities of $h_{uv}$ being synergistic, antagonistic, or null as the posterior probability that $\gamma_{uv} > 0$, $\gamma_{uv} < 0$ or $\gamma_{uv} = 0$, respectively, and classify $h_{uv}$ to the category with the highest posterior probability. 


\subsection{Two exposures}
\label{sec:p-2-simulations}


We first carry out simulation experiments with $p=2$ exposures. We consider four different scenarios and evaluate the competitors on test set RMSE in estimating $H$ and the pairwise interaction $h_{12}$. In all the scenarios, we train the methods on a sample of size $n=500$ and evaluate on $n_t = 500$ test points. The pairwise interaction $h_{12}$ is taken to be of the form $h_{12}(x_1, x_2) = \gamma_0 f(x_1, x_2)$, where $\gamma_0$ is the interaction strength and $f$ is a baseline interaction function which we vary across the four different scenarios.
We vary $\gamma_0 \in \{1,2\}$ and $\sigma_0^2 \in \{0.1,0.5\}$, thereby considering low-to-moderate interaction strength and error variance. 
For each method in a particular scenario and a pair $(\gamma_0, \sigma_0^2)$, we replicate the experiment $R = 100$ times and report the average error across replicates. When fitting SAID, we assume that the main effects satisfy the origin constraint.

We first consider a scenario where the true interaction is synergistic and nonlinear in nature and denote this case by SN. To elaborate, the dose response surface $H(\cdot)$ is given by
$$H(x_1, x_2) = 0.5 + x_1^2 + x_2^2 + \gamma_0 x_1^2 x_2^2.$$
The pairwise interaction $h_{12}(x_1, x_2)$ in this case is $h_{12}(x_1, x_2) = \gamma_0 x_1^2 x_2^2$. 
The RMSEs of each method for varying signal and noise components are provided in Tables \ref{tab:p_2_case1_whole} and \ref{tab:p_2_case1_int}. 
Although methods such as BKMR and RAMP have similar out-of-sample RMSEs with respect to SAID, the proposed approach shows superior performance in estimating the interaction term. In particular, the gains of using SAID are the most prominent when the signal strength of the interaction is the highest, indicating lack of flexibility of the other methods. 
\begin{table}
\centering
\begin{tabular}{|l|l|l|l|l|l|l|l|}
\hline
    Signal and Noise   & SAID  & BKMR & MixSelect & HierNet & Family & PIE  & RAMP \\ \hline
$ \gamma_0 = 1, \sigma_0^2 = 0.1$ & 0.05 & 0.05 &   0.15        & 0.22   & 0.41   & 1.35 &   0.05   \\ \hline
$\gamma_0 = 1, \sigma_0^2 = 0.5$ & 0.10 & 0.11 &   0.16        & 0.22    & 0.41   & 1.07 &  0.10 \\ \hline
$\gamma_0 = 2, \sigma_0^2 = 0.1$ & 0.05 & 0.05 &    0.18       & 0.25    & 0.52   & 2.10 & 0.07     \\ \hline
$\gamma_0 = 2, \sigma_0^2 = 0.5$ & 0.10 & 0.11 &  0.19         & 0.25    & 0.52   & 1.99 & 0.10  \\ \hline
\end{tabular}
\caption{RMSE of competitors when estimating H in case SN. Here $n = 500$ and $p = 2$. Lower values are better.}
\label{tab:p_2_case1_whole}
\end{table}
\begin{table}
\centering
\begin{tabular}{|l|l|l|l|l|l|l|l|}
\hline
  Signal and Noise    & SAID   & MixSelect & HierNet & Family & PIE  & RAMP \\ \hline
$\gamma_0 = 1, \sigma_0^2 = 0.1$ & 0.06 &   0.14        &  0.19   & 0.18  & 0.15  &    0.15  \\ \hline
$\gamma_0 = 1, \sigma_0^2 = 0.5$ & 0.12 &   0.14        & 0.18    & 0.18   & 0.21 &  0.19    \\ \hline
$\gamma_0 = 2, \sigma_0^2 = 0.1$ & 0.07  &    0.27       & 0.20    & 0.37   & 0.31 & 0.31     \\ \hline
$\gamma_0 = 2, \sigma_0^2 = 0.5$ & 0.14 &    0.25       & 0.22    & 0.37   & 0.31 & 0.31  \\ \hline
\end{tabular}
\caption{RMSE of competitors when estimating $h_{12}$ in case SN. Here $n = 500$ and $p = 2$. Lower values are better.}
\label{tab:p_2_case1_int}
\end{table}

Secondly, we assume the data generating process has linear main effects and a synergistic linear interaction, namely a quadratic regression (QR) setup. We let $H$ be
\begin{equation*}
    H(x_1, x_2) = 0.5 + x_1 + x_2 + \gamma_0 x_1 x_2.
\end{equation*}
The interaction is given by $h_{12}(x_1, x_2) = \gamma_0 x_1 x_2$. We next consider a scenario where the interaction is nonlinear and is neither synergistic, antagonistic, or null, while keeping the main effects as before. We call this the mis-specified interaction (MIS) case and let $H$ be
\begin{equation*}
    H(x_1, x_2) = 0.5 + x_1 + x_2 + \gamma_0 (x_1 x_2 - 2x_1^2 x_2^2).
\end{equation*}
In this case, the interaction is given by 
$h_{12}(x_1,x_2) = \gamma_0 (x_1 x_2 - 2 x_1^2 x_2^2).$
The interaction $h_{12}$ is non-negative for $x_1 x_2 \leq 1/2$ and non-positive for $x_1 x_2 \geq 1/2$, and thus is neither synergistic, antagonistic, or null.
To ease exposition, we defer the results obtained from scenarios QR and MIS to 
Section \ref{supp:p_2_table} of the  Supplementary Material. For scenario MIS, we also provide the proportion of replicates for which $h_{12}$ was misclassified as null for all the methods considered. In the scenario QR, quadratic regression approaches such as RAMP and MixSelect perform better as the data are generated from a quadratic regression. However, the performance of SAID is similar to its competitors for estimating both the dose response surface and the interaction surface. 
In the scenario MIS, BKMR and SAID perform the best in terms of estimating $H$. In terms of variable selection, performance is comparable between the methods, with MixSelect performing the best in the lowest SNR case. SAID outperforms the competitors by a clear margin in terms of estimating the interaction effects due to the lack of flexibility for other pairwise interaction approaches.

Finally, we also consider a simulation example where it is not immediate that the interaction $h_{12}$ can be decomposed as $h_{12}(x_1, x_2) = P_1(x_1) P_2(x_2) - N_1(x_1) N_2(x_2)$ for some non-negative functions $P_1, P_2, N_1, N_2$, with $x_1, x_2 \in [0,1].$ Although the proposed approach is not tailored towards detecting such interactions, it may be easily extended to incorporate tensor product splines as the basis functions with appropriately chosen prior distributions for the coefficients. For further details on such a generalized SAID (G-SAID), we refer the reader to Section \ref{supp:other} of the Supplementary Material. We let $h_{12} = \gamma_0 \sin^2(2 \pi x_1 x_2)$ for $x_1, x_2 \in [0,1]$, 
and denote this simulation example as SINE. We assume that the true intercept and both of the true main effects are $0$, so that the dose response function $H(\mathbf{x}) = h_{12}(x_1,x_2)$ in this case. The results are deferred to Table \ref{tab:p_2_sin_whole} of the Supplementary Material. Due to misspecification, SAID performs worse than BKMR in most of the cases; however, with low SNR ($\approx 0.06$) for the case $\gamma_0 = 1$ and $\sigma_0^2 = 0.5$, SAID outperforms the overly flexible BKMR, indicating superior performance of SAID in detecting interactions in the presence of noise. This is notable as most studies regarding health effects of environmental mixtures have low SNR. G-SAID performs the best throughout all the cases. 

\subsection{More than two exposures}
\label{sec:multi-exposures-simulations}


In this subsection, we consider a moderate dimensional example with $p = 10$ and thus $\binom{10}{2} = 45$ pairwise interactions. We assume that the true data generating model has $5$ synergistic, $5$ antagonistic, and $35$ null pairwise interactions. We vary $\sigma_0^2 \in \{0.2, 0.5\}$ with the number of training and test points taken to be $n = 1000$ and $n_t = 500$, respectively. The experiment is replicated $R = 100$ times for each choice of $\sigma_0^2$.
We fit all the quadratic regression methods assuming weak heredity of the pairwise interactions. Family is not considered due to unstable estimates; furthermore, except when estimating $H$, we do not consider BKMR.

We assume \eqref{eq:model-simulation} and decompose $H({\bf x})$ as
$$H({\bf x}) = \alpha_0 + M_0({\bf x}) + S_0({\bf x}) + A_0({\bf x}),$$
where
\begin{align*}
\alpha_0 &= 
-5/6\\
    M_0({\bf x}) &= \left(x_1 + x_1^2\right) + \dfrac{x_2}{2} + x_7^3,\\
    S_0({\bf x}) &= 4(x_1 - x_1^2)x_2 + x_1 x_9 + x_2^2 x_3^2 + x_3 x_8 + \dfrac{(e^{x_5} - 1) x_{10}}{e - 1},\\
    A_0({\bf x}) &= -\left[x_1 x_3 + x_2^2 x_5 + \dfrac{27}{4}x_4^2(1-x_4)x_9 + x_7 x_{10} + x_8 x_9^2\right].
\end{align*}
Here, $M_0$, $S_0$, $A_0$ denote the true main effects, synergistic interaction effects, and the antagonistic interaction effects, respectively. The form of the interactions include pairwise linear, nonlinear polynomial, and nonlinear interactions which are not of polynomial form. Each pairwise interaction has absolute maximum value equal to $1$. The signal-to-noise ratios (SNRs) for estimating the overall surface $H({\bf x})$ and the interaction surface $I(\mathbf{x}) = S_0(\mathbf{x}) + A_0(\mathbf{x})$ are $2.90$ and $1.03$, respectively, when $\sigma_0^2 = 0.2$, and $1.16$ and $0.41$, respectively, when $\sigma_0^2 = 0.5$. As before, we assume the main effects start from the origin when fitting SAID. The results comparing the methods are given in Table \ref{tab:multi-exposure-simulations} for $\sigma_0^2 = 0.2$ and Table \ref{tab:multi-exposure-simulations-case2} for $\sigma_0^2 = 0.5$. For comparison, the RMSE when estimating $H$ using BKMR is 0.17 when $\sigma_0^2 = 0.2$ and $0.25$ when $\sigma_0^2 = 0.5$.

 In terms of estimation accuracy, SAID performs uniformly better than its competitors in estimating both the overall surface $H$ and the interaction surface $I$. We believe that this is due to the SAID prior on the pairwise interactions being flexible enough to capture nonlinear interactions while also efficiently extracting interaction signal in the presence of noise. Furthermore, the SAID framework does not require any heredity assumptions, helping detection of pairwise interactions even in the absence of main effects of one or both of the corresponding exposures. In terms of variable selection, we have reported all probabilities rounded up to the second place of decimal for ease of exposition. The case 1 error probabilities for synergistic and antagonistic interactions and case 2 error probabilities for detecting null interactions are similar across the methods. However, SAID has considerably lower case 1 error when detecting null interactions, indicating that SAID more accurately classifies null interactions to be null compared with its competitors. SAID also shows superior performance in terms of case 2 error when classifying both synergistic and antagonistic interactions. This indicates that SAID has a lower probability of misclassifying an interaction as synergistic or antagonistic when it is not so, compared with its competitors. For the case with $\sigma_0^2 = 0.2$, out of the methods considered, SAID is the only one with case 1 and case 2 classification errors less than $0.05$ in each scenario. For the low SNR case with $\sigma_0^2 = 0.5$, SAID outperforms other approaches in estimation error while stacking up similarly to the other methods as in the $\sigma_0^2 = 0.2$ case.    
 When considering   uncertainty quantification of interactions evaluated at the test points, the $95\%$ posterior credible intervals obtained from SAID have $\sim 95 \%$ coverage, averaged over $R = 100$ replicates and all $45$ interactions, for both $\sigma_0^2 = 0.2$ and $\sigma_0^2 = 0.5$.

\begin{table}[]
\centering
\begin{tabular}{|l|l|l|l|l|l|l|}
\hline
            & SAID  & MixSelect & HierNet & PIE  & RAMP \\ \hline
Overall Surface RMSE     & 0.14 &   0.27     & 0.31    & 1.60 & 0.18 \\ \hline
Interaction Surface RMSE & 0.21 &    0.68    & 0.47   & 0.44 & 0.44 \\ \hline
Synergistic Case 1 Probability  & 0.01 &    0.03      & 0.01    & 0.01 & 0.01 \\ \hline
Synergistic Case2 Probability   & 0.01 &    0.26             & 0.46    & 0.21 & 0.24 \\ \hline
Antagonistic Case1 Probability   & 0.01 &     0.01           & 0.01    & 0.01 & 0.01 \\ \hline
Antagonistic Case2 Probability   & 0.01 &      0.05           & 0.26    & 0.04 & 0.07 \\ \hline
Null Case1 Probability  & 0.01 &      0.15          & 0.36    & 0.12 & 0.16 \\ \hline
Null Case2 Probability  & 0.01 &      0.04           & 0.01    & 0.02 & 0.03 \\ \hline
\end{tabular}
\caption{Comparison of the methods in terms of RMSE and variable selection accuracy for $p=10$ when $\sigma_0^2 = 0.2$. Lower values are better.
}
\label{tab:multi-exposure-simulations}
\end{table}

\begin{table}[]
\centering
\begin{tabular}{|l|l|l|l|l|l|l|}
\hline
            & SAID  & MixSelect & HierNet & PIE  & RAMP \\ \hline
Overall Surface RMSE     & 0.22 &   0.29     & 0.31    & 1.72 & 0.25 \\ \hline
Interaction Surface RMSE & 0.34 &    0.73    & 0.48   & 0.59 & 0.55 \\ \hline
Synergistic Case 1 Probability  & 0.03 &    0.03      & 0.01    & 0.01 & 0.01 \\ \hline
Synergistic Case2 Probability   & 0.12 &    0.40             & 0.53    & 0.36 & 0.51 \\ \hline
Antagonistic Case1 Probability   & 0.01 &     0.01           & 0.01    & 0.01 & 0.01 \\ \hline
Antagonistic Case2 Probability   & 0.21 &      0.32           & 0.33    & 0.25 & 0.43 \\ \hline
Null Case1 Probability  & 0.16 &      0.36          & 0.43    & 0.30 & 0.47 \\ \hline
Null Case2 Probability  & 0.04 &      0.05           & 0.01    & 0.02 & 0.02 \\ \hline
\end{tabular}
\caption{Comparison of the methods in terms of RMSE and variable selection accuracy for $p=10$ when $\sigma_0^2 = 0.5$. Lower values are better.
}
\label{tab:multi-exposure-simulations-case2}
\end{table}

\section{Analysis of Kidney Function Data}
\label{sec:SIMapplication}

\subsection{Preliminaries}
\label{sec:appprelim}

In this Section, we apply the proposed Synergistic Antagonistic Interaction Detection (SAID) approach to the NHANES 2015-16 data. We are interested in detecting synergistic and antagonistic interactions between heavy metals affecting kidney function. Following Section \ref{sec:application} and the results in Section \ref{sec:simulation}, it is evident that the current methods either do not provide inferences on synergistic or antagonistic interactions or lack flexibility in characterizing such interactions. We now illustrate how SAID detects synergistic, antagonistic, and null interactions. 

As discussed in Section \ref{sec:motivation}, we assess kidney function of individuals through their urine creatinine (uCr) levels, measured in mg/dL. Heavy metal concentrations are also measured in urine. We consider $13$ heavy metals, namely Antimony (Sb), Barium (Ba),  Cadmium (Cd), Cesium (Cs), Cobalt (Co), Lead (Pb), Manganese (Mn),  Molybdenum (Mo), Strontium (Sr), Thallium (Tl), Tin (Sn), Tungsten (W), and Uranium (U), all measured in $\mu$g/mL. Following Section \ref{sec:data-description}, we remove individuals with albuminuria and missing entries from the original data set. The sample size after removal of such entries is $n = 1979$. As described in Section \ref{sec:urine-dilution}, we multiply both the uCr and the heavy metal concentration levels by the individual-specific urine flow rate to obtain their dilution-adjusted versions. Furthermore, we also consider age (in years), sex (0 for males and 1 for females), ethnicity (Non-Hispanic White, Non-Hispanic Black, Mexican American, Other Hispanic, and Other), and body mass index (BMI) of the subject as covariates.

Before beginning our analysis, we (natural) log-transformed the dilution-adjusted urine creatinine levels. We use a marginal cumulative distribution function (CDF) transformation to 
obtained standardized exposure variables $x_j = \hat{F}_j(E_j) \in [0,1]$, with $F_j$ the marginal CDF of the $j$th exposure $E_j$ for $j=1,\ldots,p$. The estimated CDF $\hat{F}_j$ is obtained from a kernel density estimate of the marginal density of $E_j$.
Standardizing exposures in this manner facilitates statistical inferences by reducing the tendency for the exposure data to be unevenly distributed, with very sparse observations for certain ranges of exposure. Transforming the exposures does not complicate interpretation, as we can convert dose response surfaces back to the original units. In addition, the transformed exposures are directly interpretable as quantiles of the exposure distribution in the sample; for example, $x_j=0.5$ reflects a median value for the $j$th exposure. We also tried other transformations of the raw exposures, such as $g(E) = E / (1 + E)$ for $E \geq 0$, but found the quantile transform to provide slightly better consistency of SAID results across simulation replications,
while facilitating interpretation.

Transforming the exposure data does not change the nature of the detected interactions. To see this, suppose the estimated interaction between the transformed exposures $x_u$ and $x_v$ is $h_{uv}(x_u, x_v)$. In terms of the original exposures, we can rewrite this interaction as $h_{uv}(x_u, x_v) = h_{uv}(\hat{F}_u(E_u), \hat{F}_v(E_v)) = h_{uv}^{*}(E_u, E_v).$ Thus, an estimated synergistic interaction $h_{uv}$ between transformed exposures $x_u$ and $x_v$ implies a synergistic interaction $h_{uv}^{*}$ between the raw exposures $E_u$ and $E_v$. Our reasoning analogously extends to antagonistic or null interactions. In our analyses, we report the main effects and interaction effects in terms of the raw exposures.

For the $i$-th individual in the data set, let $y_i \in \mathbb{R}$ be the log of the dilution-adjusted urinary creatinine level, $\mathbf{x}_i \in [0,1]^{13}$ be the $13$ dilution-adjusted urinary metal concentrations after CDF transformations,  and $\mathbf{z}_i \in \mathbb{R}^7$ be the covariate vector, with dummy variables for
the different ethnic groups. We use ``Non-Hispanic White'' as the baseline category in defining indicators. The dimensions of exposures and covariates are $p = 13$ and $q = 7$, respectively. We also standardized the covariates age and BMI to have variance $1$ before fitting the model.


Following Section \ref{sec:posterior-sampling}, we employed a Hamiltonian Monte Carlo (HMC)-within-Gibbs sampler to obtain posterior samples of the model parameters. We ran the sampler for a total of $25000$ iterations and discarded the first $10000$ iterations as burn-in. We standardized $y$ to have variance $1$ before fitting the model. 
To compute the PIPs of the interactions, we employ the S2M approach as described in Section \ref{sec:variable-selection}. This approach detected 4 interactions with PIPs greater than $0.75$. We present our inferences in the original scale, that is, with $y$ not standardized. This is achieved simply by multiplying the estimates of the intercept, main effects, interaction effects, covariate effects, and measurement error standard deviation obtained from the fitted model by the standard deviation of $y$. 

\subsection{Model Diagnostics}
\label{sec:appdiag}

We first assess MCMC convergence. To improve mixing, we perturb the step-size $e_0$ by a small factor every $500$ iterations. As a measure of mixing of the chain, we look at the MCMC samples of the error variance $\sigma^2$ and the intercept $\alpha$. The effective sample sizes of $\sigma^2$ and $\alpha$ are approximately $15 \%$ and $20 \%$ of the total number of MCMC samples, respectively, after discarding the burn-in. We found this proportion to remain fairly robust across longer chains. To assess overall convergence of our algorithm, we look at the trace plots of $\sigma^2$ and $\alpha$, which indicated satisfactory convergence. To inspect convergence of the interactions $h_{uv}$, we use the diagnostic described in \cite{geweke1991evaluating} on the MCMC samples of $\int h_{uv}^{+}$ and $\int h_{uv}^{-}$ for $1 \leq u < v \leq p$, and obtained convergence for all the interactions. On a MacBook Pro with M1 Pro CPU and 32 GB of RAM, it took $\sim 30$ minutes to fit SAID on this dataset; for comparison, MixSelect took $\sim 7$ hours with the same number of MCMC iterates. 

Next, we assess goodness-of-fit by inspecting standardized residuals and carrying out posterior predictive checks \citep{gelman1996posterior}. We generate posterior predictive samples corresponding to each training point. To validate the assumption of normality of the measurement error, we look at a Q-Q plot of standardized residuals, defined to be the difference of the observed and predicted responses and then standardized.  We also compare the marginal density of the observed response variable with the marginal densities of $3$ randomly chosen MCMC samples of predicted responses. We provide both the Q-Q plot of the obtained standardized residuals and the comparison of marginal density plots in Figure \ref{fig:gofit}. Figure \ref{fig:gofit}(a) indicates that the normality assumption is justified; furthermore, the densities of the randomly chosen predicted response samples in Figure \ref{fig:gofit}(b) very closely resemble the marginal density of the observed response. Lastly, the coverage of $95\%$ posterior predictive intervals of the responses, averaged over all observed responses, is approximately $96 \%$.

\begin{figure}
\begin{tabular}{cc}
  \includegraphics[width=65mm]{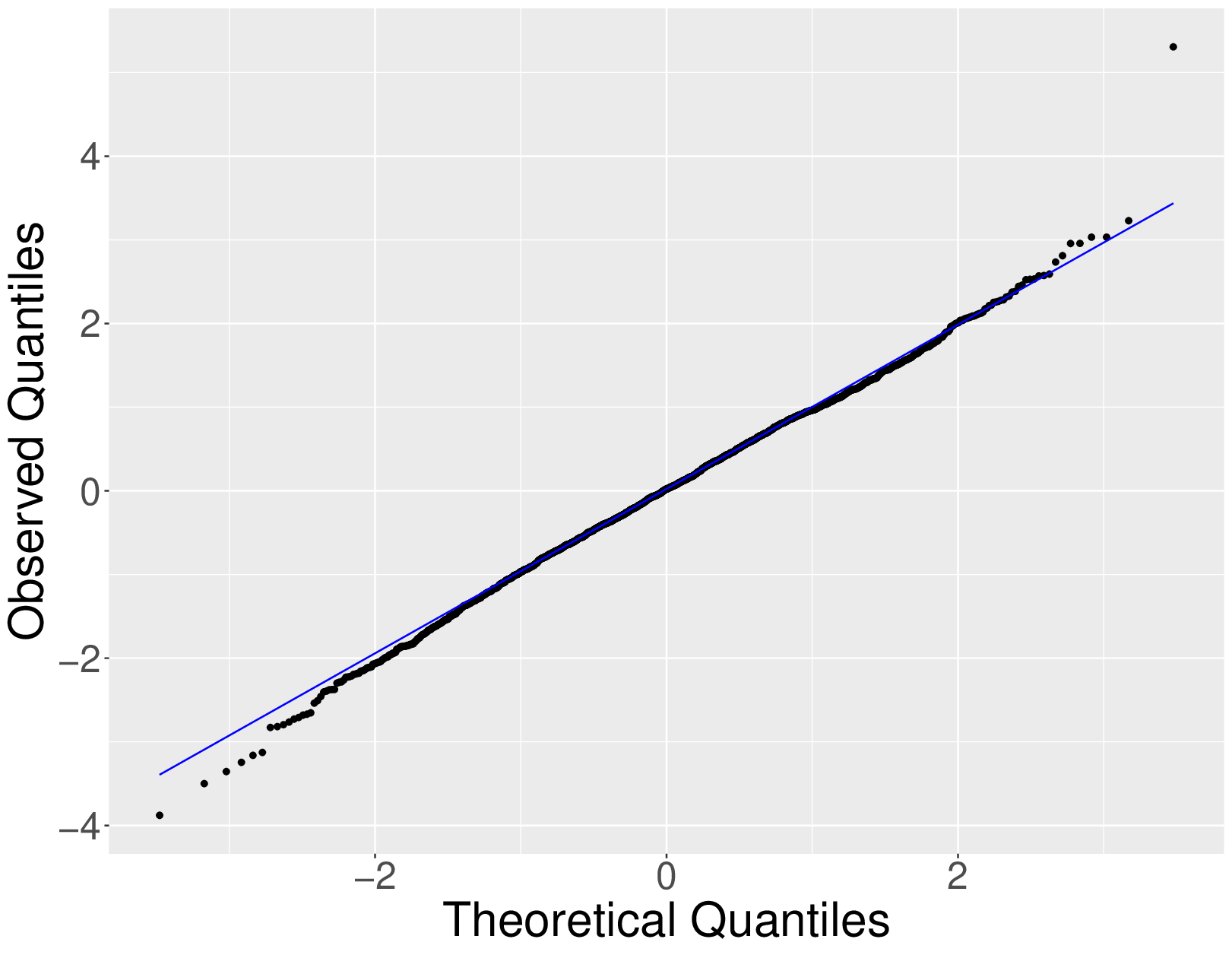} &   \includegraphics[width=65mm]{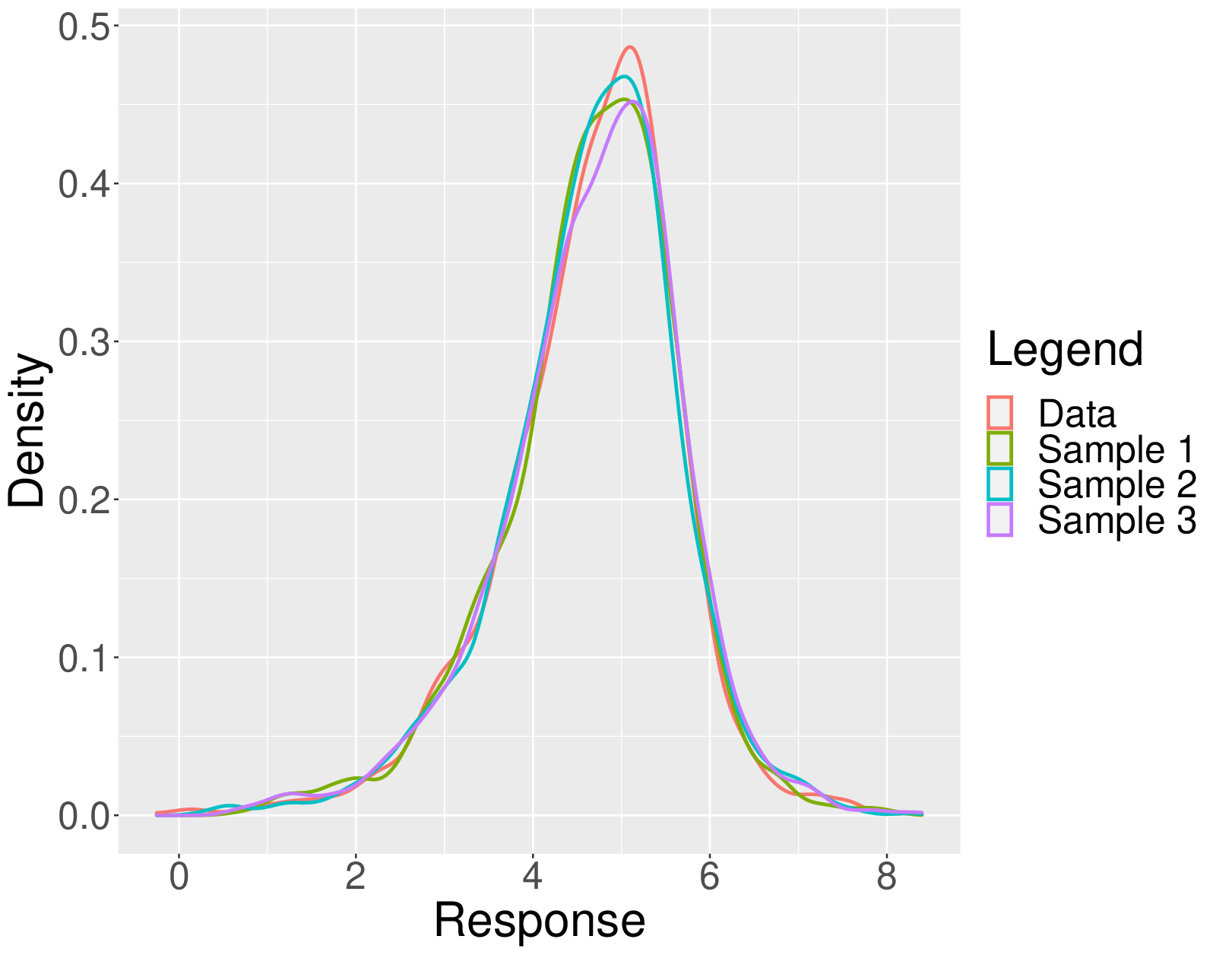} \\
(a) Q-Q Plot of standardized residuals. & (b) Marginal density plot of generated responses. \\[6pt]
\end{tabular}
\caption{Figure showing Q-Q plot of standardized residuals and marginal density plots of posterior predictive samples of urine creatinine levels, obtained from the NHANES 2015-16 data. Deviation from the reference blue line in Q-Q plot is deviation from normality. }
\label{fig:gofit}
\end{figure}



\subsection{Results}
\label{sec:appresults}

In this subsection, we discuss the main results of our analyses of the data described in Section \ref{sec:appprelim}. For each exposure, we constrain the main effects to start from $0$ at the minimum observed value of that exposure. 
The $95 \%$ posterior credible interval of the error variance $\sigma^2$ is $[0.080, 0.091]$ with a posterior mean of $0.086$. The exposures and covariates together explain approximately $89 \%$ of the variation in the response.

We observed nonlinear main effects for Antimony, Cadmium, Cobalt, Cesium, Molybdenum, Strontium, Tin, and Uranium. The general trend for the main effects of these heavy metals is an increase in log dilution-adjusted urine creatinine as exposure to metal concentrations increases. Since urine concentrations of creatinine are directly correlated with serum concentrations of creatinine, this might indicate higher kidney stress at higher levels of metal exposure. The main effects usually exhibit a monotonic or hill-shaped pattern, which is as expected in studying health effects of potentially toxic exposures.
We found exposure to high doses of Cesium to increase log dilution-adjusted urine creatinine the most, followed by Cadmium.
For illustration, we plot the main effects of Cadmium, Cesium, Molybdenum, and Uranium in Figure \ref{fig:me-plots}, with both the response and the exposures shown in log scale. We provide further plots for the main effects of the other exposures in Section \ref{supp:application} of the Supplementary Material.
When the exposures are CDF transformed, we can evaluate the main effect of an exposure at a desired quantile. As an illustration, the metals Antimony, Cadmium, Cobalt, Cesium, Molybdenum, Strontium, Tin, and Uranium at their median exposure levels have a main effect of $0.17, 0.76, 0.37, 0.92, 0.38$, $0.14$, $0.08$, and $0.25$, respectively, on log dilution-adjusted Creatinine. 
Similar main effects of heavy metal exposures on kidney function have been detected in earlier literature. 
For example, 
\cite{ferraro2010low} found Cadmium exposure to be associated with increased risk of chronic kidney disease, based on an analysis of NHANES data from 1999-2006. In a recent study, \cite{rahman2022association} found Cesium, Cadmium, and Antimony to be associated with kidney damage.  

\begin{figure}
\centering
\begin{subfigure}{20em}
    \centering
    \includegraphics[width=20em]{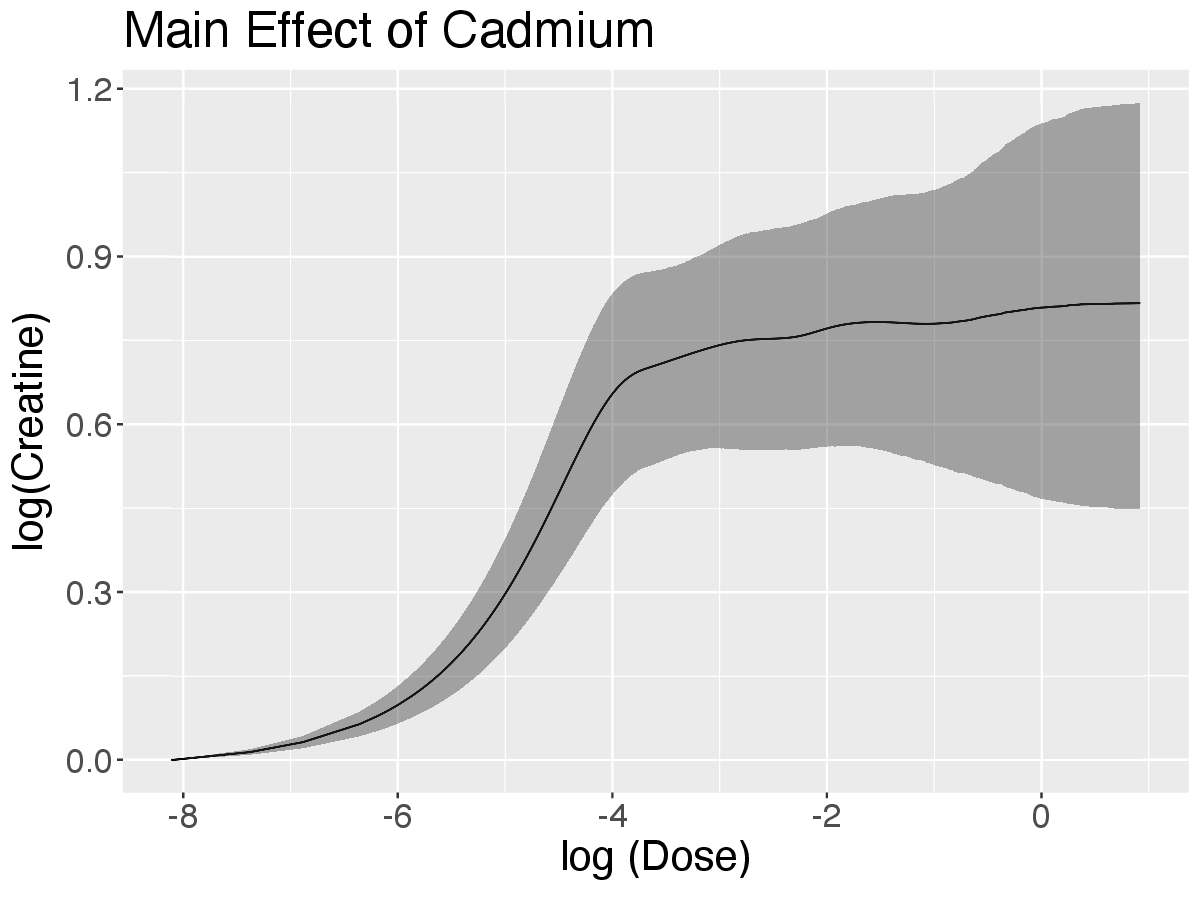}
\end{subfigure}%
\begin{subfigure}{20em}
    \centering
    \includegraphics[width=20em]{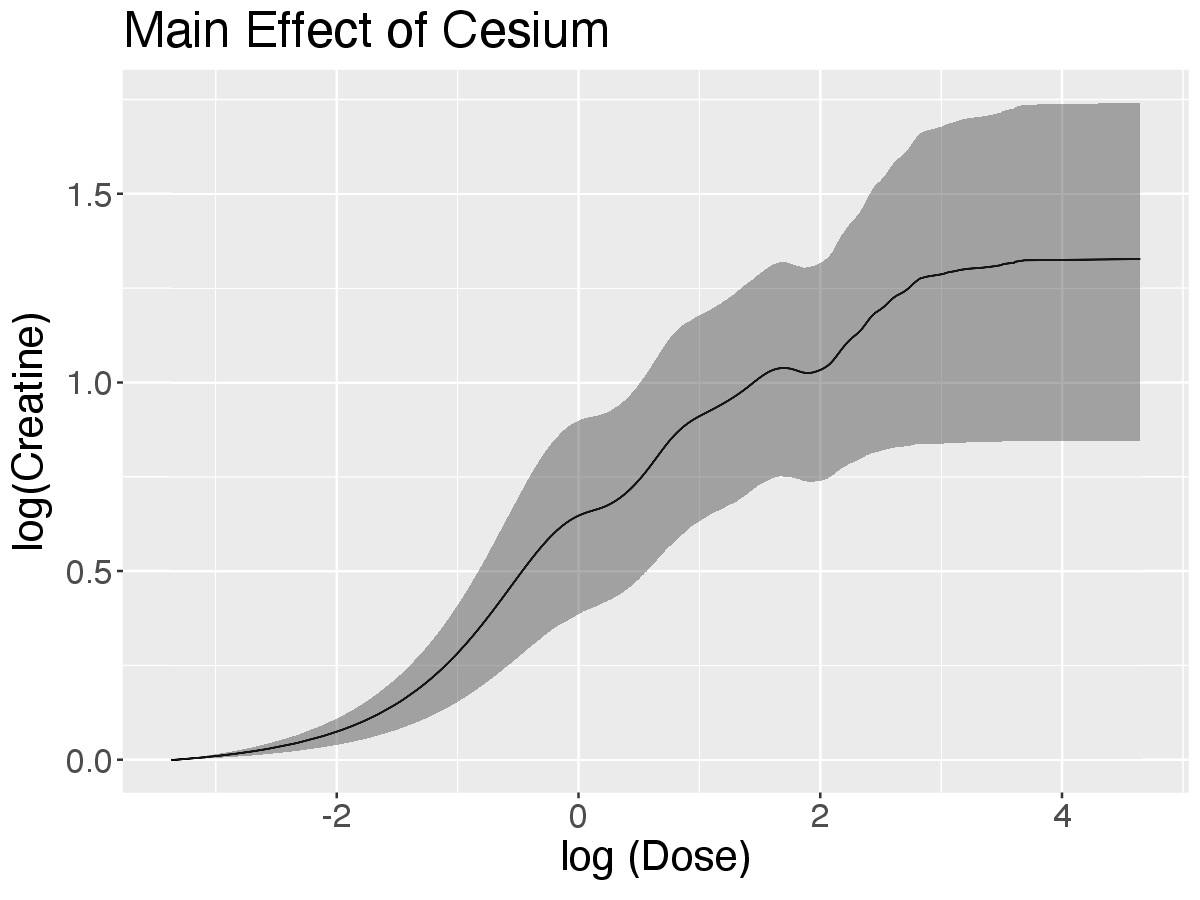}
\end{subfigure}
\begin{subfigure}{20em}
    \centering
    \includegraphics[width=20em]{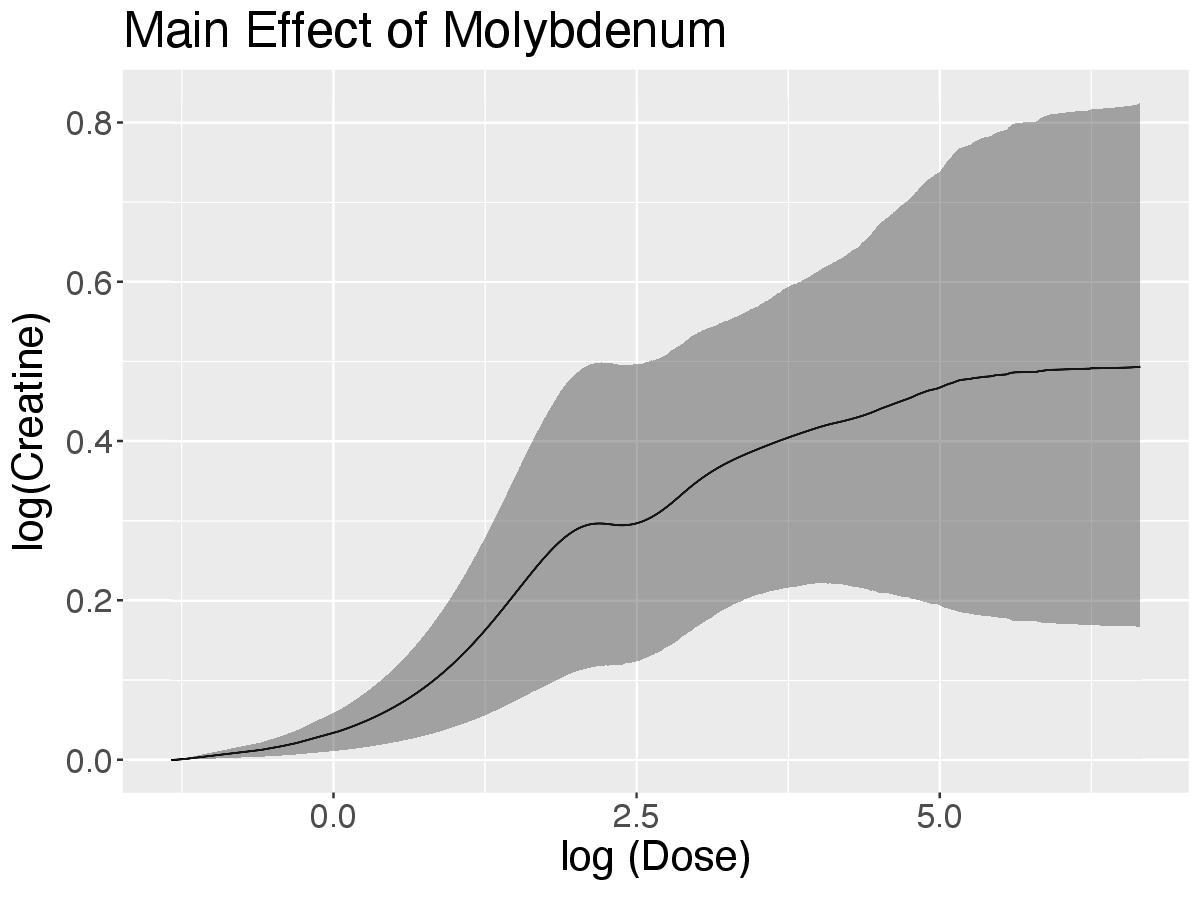}
\end{subfigure}%
\begin{subfigure}{20em}
    \centering
    \includegraphics[width=20em]{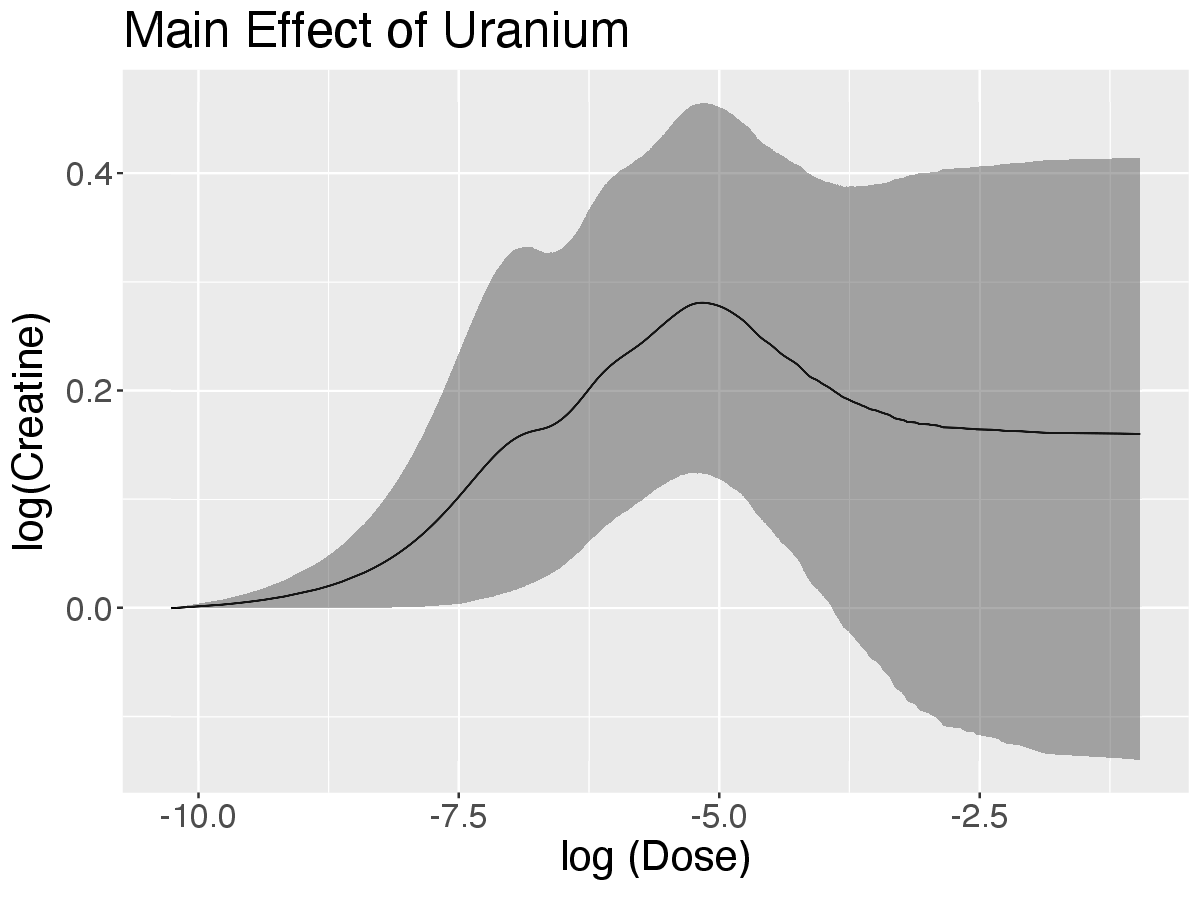}
\end{subfigure}
\caption{Plots showing the main effects of dilution-adjusted Cadmium, Cesium, Molybdenum, and Uranium on log dilution-adjusted Creatinine. Exposure levels are in log scale. Black line denotes posterior mean and shaded regions denote pointwise $95 \%$ posterior credible intervals. }
\label{fig:me-plots}
\end{figure}

The proposed approach detects multiple synergistic and antagonistic interactions between exposures. Excluding the interactions with $\mbox{PIP} < 0.75$, the detected interactions were Cadmium and Tin (PIP > 0.999), Cadmium and Manganese (PIP = 0.93), Cobalt and Manganese (PIP = 0.89), and Antimony and Uranium (PIP = 0.84).
For these interactions, we also compute the posterior synergistic probability (PSP) and posterior antagonistic probability (PAP).
The interaction between Cadmium and Tin had PSP > 0.99, indicating a very high posterior probability of being synergistic. The interactions between Cobalt and Manganese and Antimony and Uranium each had PSP $\sim 0.70$, indicating a moderate synergistic interaction. On the other hand, the interaction between Cadmium and Manganese is mildly antagonistic, with PAP = 0.46.
The interactions typically demonstrate flat behavior for most of the exposure domain and synergy/antagonism for the rest of the domain. Nonlinear surfaces of this kind cannot be captured by quadratic regression approaches due to their inflexibility. We believe this to be a possible reason behind MixSelect being unable to detect these interactions. 
In Figure \ref{fig:ie-plots}, we plot the synergistic interaction surfaces of Cadmium and Tin, Cobalt and Manganese, and Antimony and Uranium. The plot for the interaction between Cadmium and Manganese can be found in Section \ref{supp:application} of the  Supplementary Material. Each row in Figure \ref{fig:ie-plots}, from left to right, shows the pointwise $2.5 \%$ quantile, mean, and pointwise $97.5\%$ quantile of the posterior samples of the interaction, evaluated on a $30 \times 30$ regular grid of points across the exposure values. 

Regarding the covariate effects, we detected a negative association with age with an estimated coefficient $-0.004$ and $95\%$ CI given by $[-0.003, -0.005]$. As expected, females had significantly lower urine creatinine levels than males, with the estimated coefficient of sex being $-0.25$ with $95 \%$ CI given by $[-0.28, -0.22]$. In addition, average creatinine levels increased with BMI, with the coefficient on BMI estimated as $0.015$ with $95\%$ CI of $[0.013, 0.017]$. The ethnicity category ``Non-Hispanic Black'' was found to have higher log dilution-adjusted urine creatinine concentrations than that of the baseline category ``Non-Hispanic White", with an estimated coefficient of $0.17$ and a $95\%$ CI $[0.13, 0.21]$. The ethnicity categories ``Mexican American'', ``Other Hispanic'', and ``Other'' were found to have lower log dilution-adjusted urine creatinine concentrations than that of ``Non-Hispanic White''; their estimated coefficients are $-0.12$ with $95\%$ CI given by $[-0.16, -0.08]$, $-0.05$ with $95\%$ CI given by $[-0.10, -0.01]$, and $-0.12$ with a $95\%$ CI given by $[-0.16, -0.07]$, respectively. Similar observations have been made previously in literature; refer to \cite{james1988longitudinal}. 

\begin{figure}
\centering
\begin{subfigure}{35em}
    \centering
    \includegraphics[width=35em]{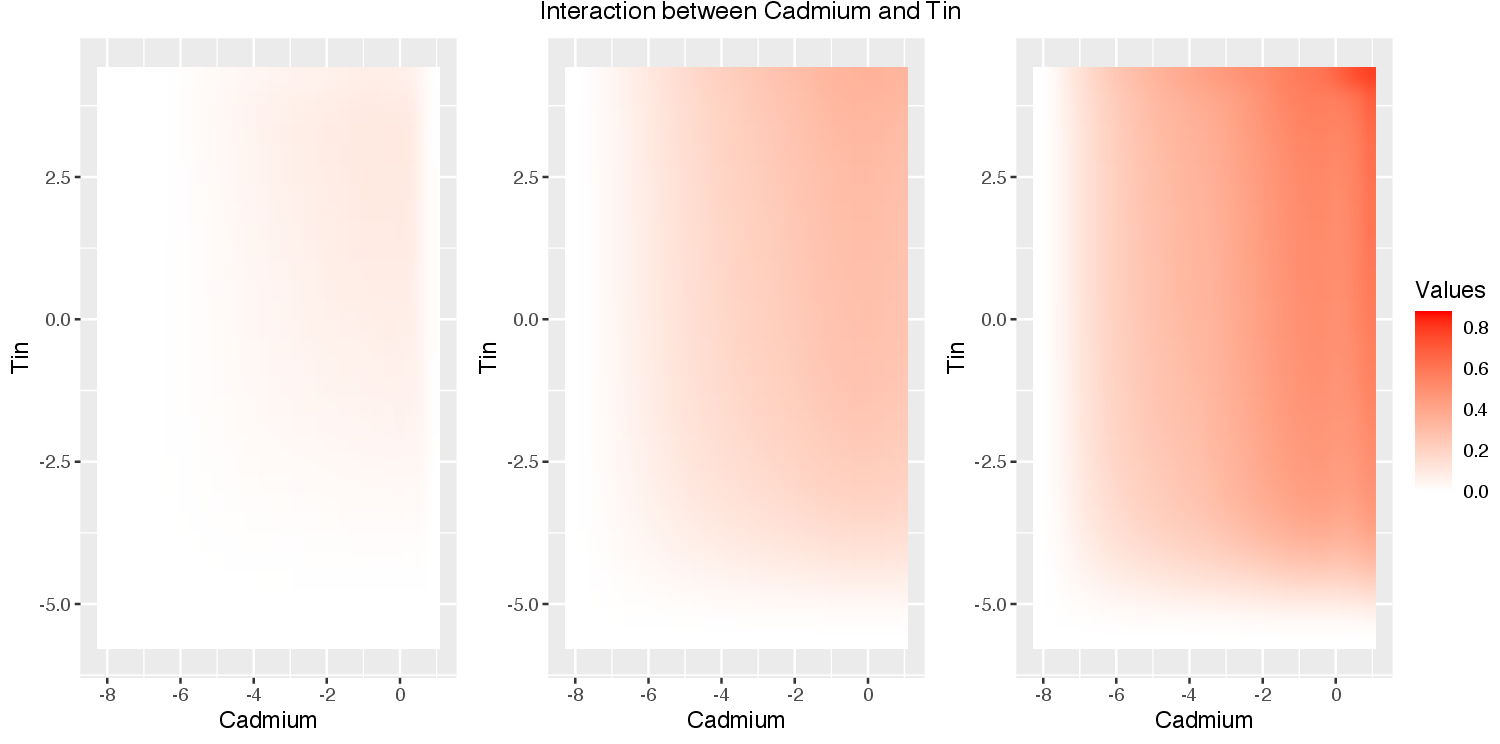}
\end{subfigure}


\begin{subfigure}{35em}
    \centering
    \includegraphics[width=35em]{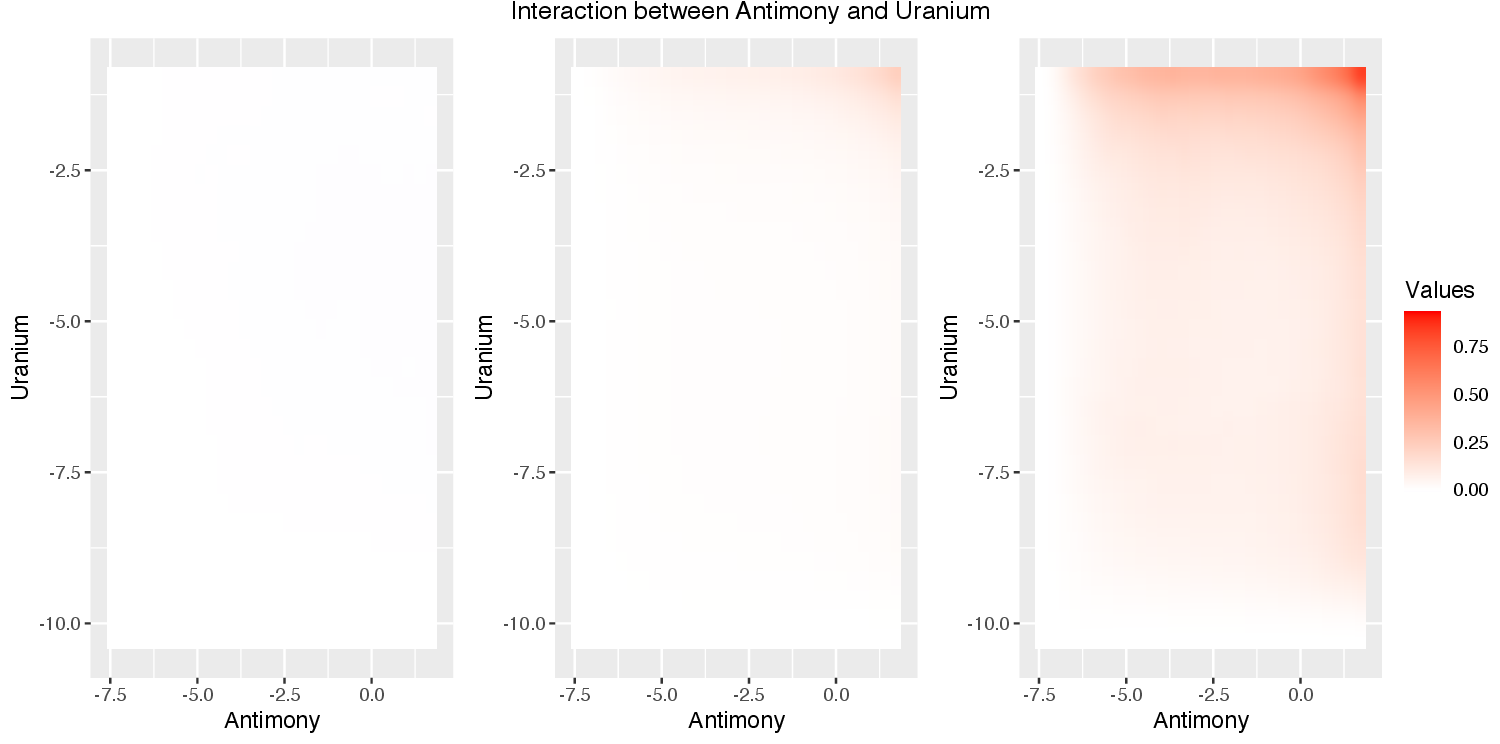}
\end{subfigure}


\begin{subfigure}{35em}
    \centering
    \includegraphics[width=35em]{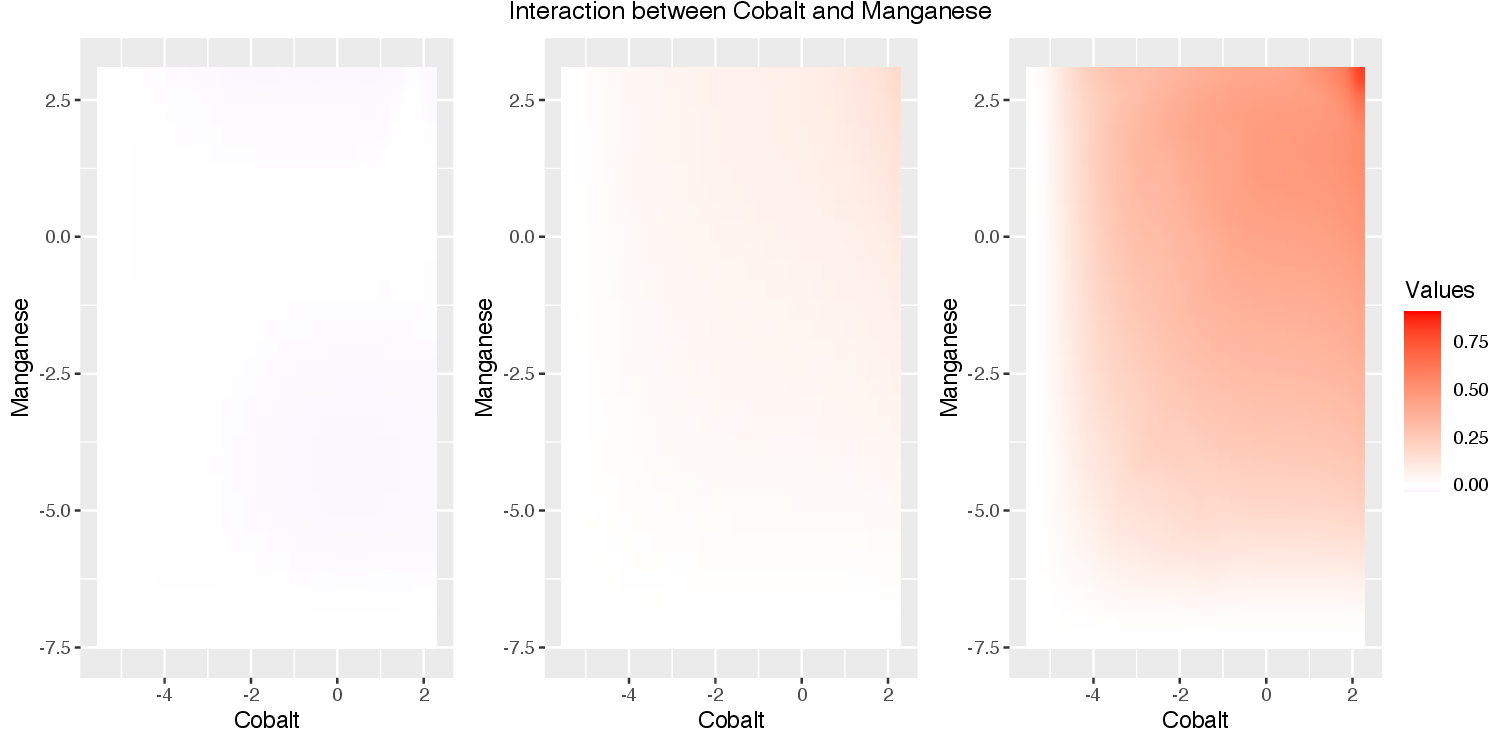}
\end{subfigure}

\caption{Plots showing the interaction effects of dilution-adjusted Cadmium and Tin, Antimony and Uranium, and Cobalt and Manganese on log dilution-adjusted  Creatinine. Exposure levels are in log scale. Each plot shows the pointwise $2.5 \%$ posterior credible surface, the posterior mean, and the $97.5 \%$ posterior credible surface from left to right.}
\label{fig:ie-plots}
   
\end{figure}

\section{Discussion}\label{sec:discussion}



In this paper, we proposed a framework for inference on main effects and pairwise synergistic, antagonistic, or null interaction effects, motivated by the {\em mixtures problem} in environmental epidemiology. 
Instead of sharply imposing constraints, we use continuous shrinkage priors to detect synergistic, antagonistic, or null interactions. When the interaction is synergistic or antagonistic, the proposed approach identifies the correct direction of interaction, while avoiding imposing exact constraints allows us to detect  interactions that are neither synergistic, antagonistic, or null. We proposed a variable selection approach to identify non-null interactions and additionally provide posterior probabilities of an interaction being synergistic and antagonistic. Unlike most quadratic regression approaches, we do not rely on heredity assumptions to detect interactions.  The proposed approach improves upon existing methods by providing greater flexibility over quadratic regression, while improving  interpretability and variable selection accuracy in detecting interactions over unrestricted nonparametric models. We demonstrated these improvements over existing state-of-the-art methods both from the perspective of estimation and variable selection.  

In our NHANES analysis, we found a variety of significant interactions among metal exposures impacting kidney function. To our knowledge, these interactions are currently unknown to the epidemiology and public health communities. It will be very interesting to validate these results in other cohorts, to study the mechanisms by which synergistic and antagonistic interactions occur in these particular metals, and also assess the regulatory implications. Our approach for assessing interactions in observational epidemiology data can be used to identify specific pairs of chemicals to study in more detail via {\em in vivo} and {\em in vitro} assays.

Although we assumed  the response to be continuous, it is straightforward to incorporate other types of responses in our framework, such as binary or count variables. In particular, we can model the response variables as falling in an exponential family with the linear predictor having exactly the form described earlier in this article. For binary outcomes and assuming a logistic link function, our proposed computational algorithm can be easily modified by relying on Polya-Gamma data augmentation \citep{polson2013polya}. A related approach can be used for counts, or our Gibbs steps can be replaced with gradient-based updates to avoid data augmentation.  It is additionally straightforward to include further applied complications, such as spatially or temporally dependent data, by incorporating appropriate terms in the linear predictor. 

The exposures in the NHANES 2015-16 data considered in this paper are mild-to-moderately correlated. In scenarios where the exposures are more heavily correlated, individually interpreting main effects and interaction effects is often infeasible. In this scenario, one could think of dividing the exposure variables into groups based on chemical class, clustering algorithms, or latent factor models. It is then of interest to investigate presence of interactions between chemical groups.

\section{Acknowledgements}
We developed the \texttt{R} package \href{https://github.com/shounakchattopadhyay/SAID}{\texttt{SAID}} 
for use in the application and numerical experiments. The package is available to download at \url{ https://github.com/shounakchattopadhyay/SAID}.
This research was supported by grants R01-ES028804, R01-ES027498, and R01-ES035625 from the National Institute of Environmental Health Sciences of the National Institutes of Health. The authors would like to thank Federico Ferrari for helpful comments.


\bibliography{ref_library}

\begin{thebibliography}{}

\bibitem[Antonelli et~al., 2020]{antonelli2020estimating}
Antonelli, J., Mazumdar, M., Bellinger, D., Christiani, D., Wright, R., and
  Coull, B. (2020).
\newblock Estimating the health effects of environmental mixtures using
  {B}ayesian semiparametric regression and sparsity inducing priors.
\newblock {\em Annals of Applied Statistics}, 14(1):257--275.

\bibitem[Barr et~al., 2005]{barr2005urinary}
Barr, D.~B., Wilder, L.~C., Caudill, S.~P., Gonzalez, A.~J., Needham, L.~L.,
  and Pirkle, J.~L. (2005).
\newblock Urinary creatinine concentrations in the us population: implications
  for urinary biologic monitoring measurements.
\newblock {\em Environmental Health Perspectives}, 113(2):192--200.

\bibitem[Baxmann et~al., 2008]{baxmann2008influence}
Baxmann, A.~C., Ahmed, M.~S., Marques, N.~C., Menon, V.~B., Pereira, A.~B.,
  Kirsztajn, G.~M., and Heilberg, I.~P. (2008).
\newblock Influence of muscle mass and physical activity on serum and urinary
  creatinine and serum cystatin c.
\newblock {\em Clinical Journal of the American Society of Nephrology},
  3(2):348--354.

\bibitem[Betancourt, 2017]{betancourt2017conceptual}
Betancourt, M. (2017).
\newblock A conceptual introduction to hamiltonian monte carlo.
\newblock {\em arXiv preprint arXiv:1701.02434}.

\bibitem[Betancourt and Girolami, 2015]{betancourt2015hamiltonian}
Betancourt, M. and Girolami, M. (2015).
\newblock Hamiltonian monte carlo for hierarchical models.
\newblock {\em Current Trends in Bayesian Methodology with Applications},
  79(30):2--4.

\bibitem[Bhattacharya et~al., 2012]{bhattacharya2012bayesian}
Bhattacharya, A., Pati, D., Pillai, N.~S., and Dunson, D.~B. (2012).
\newblock Bayesian shrinkage.
\newblock {\em arXiv preprint arXiv:1212.6088}.

\bibitem[Bien et~al., 2013]{bien2013lasso}
Bien, J., Taylor, J., and Tibshirani, R. (2013).
\newblock A lasso for hierarchical interactions.
\newblock {\em Annals of Statistics}, 41(3):1111--1141.

\bibitem[Bobb et~al., 2015]{bobb2015bayesian}
Bobb, J.~F., Valeri, L., Claus~Henn, B., Christiani, D.~C., Wright, R.~O.,
  Mazumdar, M., Godleski, J.~J., and Coull, B.~A. (2015).
\newblock Bayesian kernel machine regression for estimating the health effects
  of multi-pollutant mixtures.
\newblock {\em Biostatistics}, 16(3):493--508.

\bibitem[Brezger and Lang, 2006]{brezger2006generalized}
Brezger, A. and Lang, S. (2006).
\newblock Generalized structured additive regression based on {B}ayesian
  p-splines.
\newblock {\em Computational Statistics \& Data Analysis}, 50(4):967--991.

\bibitem[Carvalho et~al., 2009]{carvalho2009handling}
Carvalho, C.~M., Polson, N.~G., and Scott, J.~G. (2009).
\newblock Handling sparsity via the horseshoe.
\newblock In {\em Artificial Intelligence and Statistics}, pages 73--80. PMLR.

\bibitem[Chipman et~al., 2010]{chipman2010bart}
Chipman, H.~A., George, E.~I., and McCulloch, R.~E. (2010).
\newblock Bart: {B}ayesian additive regression trees.
\newblock {\em The Annals of Applied Statistics}, 4(1):266--298.

\bibitem[Czarnota et~al., 2015]{czarnota2015assessment}
Czarnota, J., Gennings, C., and Wheeler, D.~C. (2015).
\newblock Assessment of weighted quantile sum regression for modeling chemical
  mixtures and cancer risk.
\newblock {\em Cancer Informatics}, 14:CIN--S17295.

\bibitem[De~Boor, 1978]{de1978practical}
De~Boor, C. (1978).
\newblock {\em A Practical Guide to Splines}, volume~27.
\newblock Springer-Verlag New York.

\bibitem[Duane et~al., 1987]{duane1987hybrid}
Duane, S., Kennedy, A.~D., Pendleton, B.~J., and Roweth, D. (1987).
\newblock Hybrid monte carlo.
\newblock {\em Physics {L}etters B}, 195(2):216--222.

\bibitem[Ferrari and Dunson, 2020]{ferrari2020identifying}
Ferrari, F. and Dunson, D.~B. (2020).
\newblock Identifying main effects and interactions among exposures using
  gaussian processes.
\newblock {\em The Annals of Applied Statistics}, 14(4):1743--1758.

\bibitem[Ferrari and Dunson, 2021]{ferrari2021bayesian}
Ferrari, F. and Dunson, D.~B. (2021).
\newblock Bayesian factor analysis for inference on interactions.
\newblock {\em Journal of the American Statistical Association},
  116(535):1521--1532.

\bibitem[Ferraro et~al., 2010]{ferraro2010low}
Ferraro, P.~M., Costanzi, S., Naticchia, A., Sturniolo, A., and Gambaro, G.
  (2010).
\newblock Low level exposure to cadmium increases the risk of chronic kidney
  disease: analysis of the nhanes 1999-2006.
\newblock {\em BMC Public Health}, 10(1):1--8.

\bibitem[Forbes and Bruining, 1976]{forbes1976urinary}
Forbes, G. and Bruining, G.~J. (1976).
\newblock Urinary creatinine excretion and lean body mass.
\newblock {\em The American Journal of Clinical Nutrition}, 29(12):1359--1366.

\bibitem[Gelman et~al., 1996]{gelman1996posterior}
Gelman, A., Meng, X.-L., and Stern, H. (1996).
\newblock Posterior predictive assessment of model fitness via realized
  discrepancies.
\newblock {\em Statistica Sinica}, pages 733--760.

\bibitem[George and McCulloch, 1993]{george1993variable}
George, E.~I. and McCulloch, R.~E. (1993).
\newblock Variable selection via gibbs sampling.
\newblock {\em Journal of the American Statistical Association},
  88(423):881--889.

\bibitem[Geweke, 1991]{geweke1991evaluating}
Geweke, J. (1991).
\newblock Evaluating the accuracy of sampling-based approaches to the
  calculation of posterior moments.
\newblock Technical report, Federal Reserve Bank of Minneapolis.

\bibitem[Hao et~al., 2018]{hao2018model}
Hao, N., Feng, Y., and Zhang, H.~H. (2018).
\newblock Model selection for high-dimensional quadratic regression via
  regularization.
\newblock {\em Journal of the American Statistical Association},
  113(522):615--625.

\bibitem[Haris et~al., 2016]{haris2016convex}
Haris, A., Witten, D., and Simon, N. (2016).
\newblock Convex modeling of interactions with strong heredity.
\newblock {\em Journal of Computational and Graphical Statistics},
  25(4):981--1004.

\bibitem[Hastings, 1970]{hastings1970monte}
Hastings, W.~K. (1970).
\newblock Monte carlo sampling methods using markov chains and their
  applications.
\newblock {\em Biometrika}, 57(1):97--109.

\bibitem[Hays et~al., 2015]{hays2015variation}
Hays, S.~M., Aylward, L.~L., and Blount, B.~C. (2015).
\newblock Variation in urinary flow rates according to demographic
  characteristics and body mass index in nhanes: potential confounding of
  associations between health outcomes and urinary biomarker concentrations.
\newblock {\em Environmental Health Perspectives}, 123(4):293--300.

\bibitem[Herring, 2010]{herring2010nonparametric}
Herring, A.~H. (2010).
\newblock Nonparametric {B}ayes shrinkage for assessing exposures to mixtures
  subject to limits of detection.
\newblock {\em Epidemiology (Cambridge, Mass.)}, 21(Suppl 4):S71.

\bibitem[Hoffman et~al., 2014]{hoffman2014no}
Hoffman, M.~D., Gelman, A., et~al. (2014).
\newblock The no-u-turn sampler: adaptively setting path lengths in hamiltonian
  monte carlo.
\newblock {\em Journal of Machine Learning Research}, 15(1):1593--1623.

\bibitem[Ishwaran and Rao, 2005]{ishwaran2005spike}
Ishwaran, H. and Rao, J.~S. (2005).
\newblock Spike and slab variable selection: frequentist and bayesian
  strategies.
\newblock {\em Annals of Statistics}, 33(2):730--773.

\bibitem[Jain, 2016]{jain2016associated}
Jain, R. (2016).
\newblock Associated complex of urine creatinine, serum creatinine, and chronic
  kidney disease.
\newblock {\em Epidemiology (Sunnyvale)}, 6(234):2161--1165.

\bibitem[James et~al., 1988]{james1988longitudinal}
James, G.~D., Sealey, J.~E., Alderman, M., Ljungman, S., Mueller, F.~B.,
  Pecker, M.~S., and Laragh, J.~H. (1988).
\newblock A longitudinal study of urinary creatinine and creatinine clearance
  in normal subjects: race, sex, and age differences.
\newblock {\em American Journal of Hypertension}, 1(2):124--131.

\bibitem[Jeng et~al., 2021]{jeng2021clinical}
Jeng, P.-H., Huang, T.-R., Wang, C.-C., and Chen, W.-L. (2021).
\newblock Clinical relevance of urine flow rate and exposure to polycyclic
  aromatic hydrocarbons.
\newblock {\em International Journal of Environmental Research and Public
  Health}, 18(10):5372.

\bibitem[Joubert et~al., 2022]{joubert2022powering}
Joubert, B.~R., Kioumourtzoglou, M.-A., Chamberlain, T., Chen, H.~Y., Gennings,
  C., Turyk, M.~E., Miranda, M.~L., Webster, T.~F., Ensor, K.~B., Dunson,
  D.~B., et~al. (2022).
\newblock Powering research through innovative methods for mixtures in
  epidemiology (prime) program: Novel and expanded statistical methods.
\newblock {\em International Journal of Environmental Research and Public
  Health}, 19(3):1378.

\bibitem[Kashani et~al., 2020]{kashani2020creatinine}
Kashani, K., Rosner, M.~H., and Ostermann, M. (2020).
\newblock Creatinine: from physiology to clinical application.
\newblock {\em European Journal of Internal Medicine}, 72:9--14.

\bibitem[Kass and Raftery, 1995]{kass1995bayes}
Kass, R.~E. and Raftery, A.~E. (1995).
\newblock Bayes factors.
\newblock {\em Journal of the American Statistical Association},
  90(430):773--795.

\bibitem[Kim et~al., 2015]{kim2015environmental}
Kim, N.~H., Hyun, Y.~Y., Lee, K.-B., Chang, Y., Rhu, S., Oh, K.-H., and Ahn, C.
  (2015).
\newblock Environmental heavy metal exposure and chronic kidney disease in the
  general population.
\newblock {\em Journal of Korean Medical Science}, 30(3):272--277.

\bibitem[Lang and Brezger, 2004]{lang2004bayesian}
Lang, S. and Brezger, A. (2004).
\newblock Bayesian p-splines.
\newblock {\em Journal of Computational and Graphical Statistics},
  13(1):183--212.

\bibitem[Li and Pati, 2017]{li2017variable}
Li, H. and Pati, D. (2017).
\newblock Variable selection using shrinkage priors.
\newblock {\em Computational Statistics \& Data Analysis}, 107:107--119.

\bibitem[Luo and Hendryx, 2020]{luo2020metal}
Luo, J. and Hendryx, M. (2020).
\newblock Metal mixtures and kidney function: An application of machine
  learning to nhanes data.
\newblock {\em Environmental Research}, 191:110126.

\bibitem[Makalic and Schmidt, 2015]{makalic2015simple}
Makalic, E. and Schmidt, D.~F. (2015).
\newblock A simple sampler for the horseshoe estimator.
\newblock {\em IEEE Signal Processing Letters}, 23(1):179--182.

\bibitem[Metropolis et~al., 1953]{metropolis1953equation}
Metropolis, N., Rosenbluth, A.~W., Rosenbluth, M.~N., Teller, A.~H., and
  Teller, E. (1953).
\newblock Equation of state calculations by fast computing machines.
\newblock {\em The Journal of Chemical Physics}, 21(6):1087--1092.

\bibitem[Middleton et~al., 2016]{middleton2016assessing}
Middleton, D.~R., Watts, M.~J., Lark, R.~M., Milne, C.~J., and Polya, D.~A.
  (2016).
\newblock Assessing urinary flow rate, creatinine, osmolality and other
  hydration adjustment methods for urinary biomonitoring using nhanes arsenic,
  iodine, lead and cadmium data.
\newblock {\em Environmental Health}, 15(1):1--13.

\bibitem[Mitchell and Beauchamp, 1988]{mitchell1988bayesian}
Mitchell, T.~J. and Beauchamp, J.~J. (1988).
\newblock Bayesian variable selection in linear regression.
\newblock {\em Journal of the american statistical association},
  83(404):1023--1032.

\bibitem[Molitor et~al., 2010]{molitor2010bayesian}
Molitor, J., Papathomas, M., Jerrett, M., and Richardson, S. (2010).
\newblock Bayesian profile regression with an application to the national
  survey of children's health.
\newblock {\em Biostatistics}, 11(3):484--498.

\bibitem[Neal, 2011]{neal2011mcmc}
Neal, R.~M. (2011).
\newblock {\em MCMC using Hamiltonian dynamics}, chapter Chapter 5, pages
  113--162.
\newblock Chapman \& Hall / CRC.

\bibitem[Pollack et~al., 2015]{pollack2015kidney}
Pollack, A.~Z., Mumford, S.~L., Mendola, P., Perkins, N.~J., Rotman, Y.,
  Wactawski-Wende, J., and Schisterman, E.~F. (2015).
\newblock Kidney biomarkers associated with blood lead, mercury, and cadmium in
  premenopausal women: a prospective cohort study.
\newblock {\em Journal of Toxicology and Environmental Health, Part A},
  78(2):119--131.

\bibitem[Polson et~al., 2013]{polson2013polya}
Polson, N.~G., Scott, J.~G., and Windle, J. (2013).
\newblock Bayesian inference for logistic models using pólya–gamma latent
  variables.
\newblock {\em Journal of the American Statistical Association},
  108(504):1339--1349.

\bibitem[Rahman et~al., 2022]{rahman2022association}
Rahman, H.~H., Niemann, D., and Munson-McGee, S.~H. (2022).
\newblock Association of albumin to creatinine ratio with urinary arsenic and
  metal exposure: evidence from nhanes 2015--2016.
\newblock {\em International Urology and Nephrology}, 54(6):1343--1353.

\bibitem[Ramsay, 1988]{ramsay1988monotone}
Ramsay, J.~O. (1988).
\newblock {Monotone regression splines in action}.
\newblock {\em Statistical Science}, 3(4):425 -- 441.

\bibitem[Rao et~al., 2016]{rao2016data}
Rao, V., Lin, L., and Dunson, D.~B. (2016).
\newblock Data augmentation for models based on rejection sampling.
\newblock {\em Biometrika}, 103(2):319--335.

\bibitem[Rue, 2001]{rue2001fast}
Rue, H. (2001).
\newblock Fast sampling of gaussian markov random fields.
\newblock {\em Journal of the Royal Statistical Society: Series B (Statistical
  Methodology)}, 63(2):325--338.

\bibitem[Samanta and Antonelli, 2022]{samanta2022estimation}
Samanta, S. and Antonelli, J. (2022).
\newblock Estimation and false discovery control for the analysis of
  environmental mixtures.
\newblock {\em Biostatistics}, 23(4):1039--1055.

\bibitem[Shively et~al., 2009]{shively2009bayesian}
Shively, T.~S., Sager, T.~W., and Walker, S.~G. (2009).
\newblock A bayesian approach to non-parametric monotone function estimation.
\newblock {\em Journal of the Royal Statistical Society: Series B (Statistical
  Methodology)}, 71(1):159--175.

\bibitem[Wang et~al., 2019]{wang2019penalized}
Wang, C., Jiang, B., and Zhu, L. (2019).
\newblock Penalized interaction estimation for ultrahigh dimensional quadratic
  regression.
\newblock {\em arXiv preprint arXiv:1901.07147}.

\bibitem[Wei et~al., 2020]{wei2020sparse}
Wei, R., Reich, B.~J., Hoppin, J.~A., and Ghosal, S. (2020).
\newblock Sparse {B}ayesian additive nonparametric regression with application
  to health effects of pesticides mixtures.
\newblock {\em Statistica Sinica}, 30(1):55--79.

\bibitem[Williams and Rasmussen, 2006]{williams2006gaussian}
Williams, C.~K. and Rasmussen, C.~E. (2006).
\newblock {\em Gaussian Processes for Machine Learning}, volume~2.
\newblock MIT press Cambridge, MA.

\end{thebibliography}
\bibliographystyle{apalike}

\newpage

\appendix

\section*{Supplementary Material}




\section{B-spline Functions}\label{supp:bspline}

For an extensive overview of B-spline basis functions, we refer the reader to \cite{de1978practical}. As a default, we choose cubic splines for modeling both the main and interaction effects throughout the paper; however, one could vary the degree of the B-splines as required. We now describe the choice of B-splines for modeling the interaction effects and main effects in Sections 3.2 and 3.4, respectively.

For $1 \leq u < v \leq p$ and the interaction function $h_{uv}$, we let $h_{uv}$ satisfy $h_{uv}(x_u, 0) = 0$ for all $x_u \in [0,1]$ and $h_{uv}(0, x_v) = 0$ for all $x_v \in [0,1]$.
To achieve this, we first construct the set of B-spline functions $s_{u0}(x_u), \ldots, s_{um}(x_u)$ along the $u$-th dimension for $u=1,\ldots,p$, where $s_{u0}$ is the intercept spline. The intercept spline $s_{u0}$ is the only B-spline basis function satisfying $s_{u0}(0) \neq 0$. Thus, to model the interaction effects, we simply ignore the intercept spline $s_{u0}$ and only choose the B-spline functions $s_{ul}(\cdot)$ that satisfy $s_{ul}(0) = 0$, namely $s_{u1}, \ldots, s_{um}$.

In order to model the main effect $f_j$ for $j=1,\ldots,p$, we proceed as follows. Suppose we have B-spline basis functions $t_{j0}, \ldots, t_{jd}$ for the $j$th main effect, with $t_{j0}$ being the intercept spline. If the main effects satisfy the origin constraint, we simply ignore $t_{j0}$ and let the basis functions be $b_{j, l}(x_u) = t_{jl}(x_u)$ for $j=1,\ldots,p$ and $l=1,\ldots,d$. If the main effects satisfy the integral constraint, we let $b_{j, l}(x_u) = t_{jl}(x_u) - \int_{0}^{1} t_{jl}(x) \, dx$ for $j=1,\ldots,p$ and $l = 0,\ldots,d$.

\section{P-spline Priors}
\label{supp:pspline}

In order to ensure a function $g(x) = \sum_{j=1}^{m} s_j(x) \beta_j$ expressed as a linear combination of the B-splines $s_1, \ldots, s_m$ is sufficiently smooth, we follow the P-spline approach described in \cite{lang2004bayesian}. P-spline priors aim to promote smoothness by shrinking coefficients of adjacent splines towards each other. 

We use a P-spline prior of order $1$ with a slight modification to ensure propriety of the prior distribution.
The order $1$ prior described in \cite{lang2004bayesian} is obtained by considering the hierarchical model $\beta_j \mid \beta_{j-1} \sim N(\beta_{j-1}, \tau^2)$ for $j = 2, \ldots, m$, and then letting $\beta_1$ have an improper prior, given by $\pi(\beta_1) \propto 1$. However, we avoid the use of improper priors, and let $\beta_1 \sim N(0, \tau^2)$ as well. Under this modified prior specification, the joint distribution of $\mathbf{\beta} = (\beta_1, \ldots, \beta_m)^{\T}$ is given by
$$\pi(\mathbf{\beta}) \propto \exp\left[-\dfrac{1}{2 \tau^2} \left\{\beta_1^2 + \sum_{j=2}^{m} (\beta_j - \beta_{j-1})^2\right\}\right].$$
This can be rewritten as $\mathbf{\beta} \sim N(0, \tau^2 \Sigma_0)$ for a positive-definite matrix $\Sigma_0$ whose entries are known. 

\section{Model Identifiability}
\label{supp:model-id}

We state a result which ensures identifiability of the main and interaction effects in the assumed model.

\begin{proposition}
\label{prop:identifiability}
Suppose the dose response function is given by $$H(\mathbf{x}) = \alpha + \sum_{j=1}^{p} f_j(x_j) + \sum_{1 \leq u < v \leq p} h_{uv}(x_u, x_v),$$ for $\mathbf{x} = (x_1, \ldots, x_p)^{\T} \in [0,1]^p$, $p \geq 2$. Assume the following conditions hold:
\begin{enumerate}
    \item $f_1, \ldots, f_p$ together satisfy either (a) or (b), where \\
    \begin{enumerate}
        \item \textbf{Origin Constraint:} $f_j(0) = 0$ for $j = 1, \ldots, p$, or
    \item \textbf{Integral Constraint:} $\displaystyle \int_{0}^{1} f_j(x_j) \, dx_j = 0$ for $j=1,\ldots,p$. \\
    \end{enumerate}
    \item For all $1 \leq u < v \leq p$, $h_{uv}(x_u, 0) = 0$ and $h_{uv}(0, x_v) = 0$ for all $x_u, x_v \in [0,1]$.
\end{enumerate}
If two tuples $\left(\alpha, \{f_j\}_{j=1}^{p}, \{h_{uv}\}_{1 \leq u < v \leq p}\right)$ and $\left(\tilde{\alpha}, \{\tilde{f}_j\}_{j=1}^{p}, \{\tilde{h}_{uv}\}_{1 \leq u < v \leq p}\right)$ satify
$$\alpha + \sum_{j=1}^{p} f_j(x_j) + \sum_{1 \leq u < v \leq p} h_{uv}(x_u, x_v) = \tilde{\alpha} + \sum_{j=1}^{p} \tilde{f}_j(x_j) + \sum_{1 \leq u < v \leq p} \tilde{h}_{uv}(x_u, x_v),$$
for all $\mathbf{x} \in [0,1]^p$, then $\left(\alpha, \{f_j\}_{j=1}^{p}, \{h_{uv}\}_{1 \leq u < v \leq p}\right) = \left(\tilde{\alpha}, \{\tilde{f}_j\}_{j=1}^{p}, \{\tilde{h}_{uv}\}_{1 \leq u < v \leq p}\right)$.
\end{proposition}
\begin{proof}

We first prove the result for $p=2$ and then generalize to higher dimensions, under both sets of constraints on the main effects. Suppose the following holds for all $x_1, x_2 \in [0,1]$: 
    \begin{equation}
        \label{supp-eq:identity}
        \alpha + f_1(x_1) + f_2(x_2) + h_{12}(x_1,x_2) = \tilde{\alpha} + \tilde{f_1}(x_1) + \tilde{f_2}(x_2) + \tilde{h}_{12}(x_1, x_2).
    \end{equation}
\textbf{Proof for origin constraint:}  We first let $x_1 = x_2 = 0$. Due to the constraints on the main effects and interactions, we obtain $\alpha = \tilde{\alpha}$. Next, we let $x_2 = 0$ for any $x_1$. We then obtain $\alpha + f_1(x_1) = \tilde{\alpha} + \tilde{f_1}(x_1)$ for all $x_1$, which leads us to $f_1 = \tilde{f_1}$. Similarly, by letting $x_1 = 0$ for any $x_2$, we obtain $f_2 = \tilde{f_2}$. Substituting these results into \eqref{supp-eq:identity}, we find $h_{12} = \tilde{h}_{12}.$

For $p \geq 3$, we fix $1 \leq u < v \leq p$. We first let $x_j = 0$ for all $j \notin \{u, v\}$, which implies
$$\alpha + f_u(x_u) + f_v(x_v) + h_{uv}(x_u, x_v) = \tilde{\alpha} + \tilde{f}_u(x_u) + \tilde{f}_v(x_v) + \tilde{h}_{uv}(x_u, x_v).$$
The result for $p=2$ implies $f_u = \tilde{f}_u$, $f_v = \tilde{f}_v$, and $h_{uv} = \tilde{h}_{uv}$ for any $1 \leq u < v \leq p$.

\textbf{Proof for integral constraint:} We first let $x_2 = 0$, which implies $\alpha + f_1(x_1) + f_2(0) = \tilde{\alpha} + \tilde{f}_1(x_1) + \tilde{f_2}(0)$ for all $x_1 \in [0,1]$. Integrating both sides with respect to $x_1$ and using the fact that $\int_{0}^{1} f_1(x_1) \, dx_1 = \int_{0}^{1} \tilde{f}_1(x_1) \, dx_1 = 0$, we obtain $\alpha + f_2(0) = \tilde{\alpha} + \tilde{f}_2(0)$, which when substituted into the previous equation, implies that $f_1 = \tilde{f}_1$. Similarly, letting $x_1 = 0$ leads us to $f_2 = \tilde{f}_2$, which implies $\alpha = \tilde{\alpha}$ since $\alpha + f_2(0) = \tilde{\alpha} + \tilde{f}_2(0)$. Substituting these results into \eqref{supp-eq:identity} implies $h_{12} = \tilde{h}_{12}$.

For $p \geq 3$, we fix $1 \leq u < v \leq p$. We first let $x_j = 0$ for all $j \notin \{u, v\}$, which implies
$$ \rho + f_u(x_u) + f_v(x_v) + h_{uv}(x_u, x_v) = \tilde{\rho} + \tilde{f}_u(x_u) + \tilde{f}_v(x_v) + \tilde{h}_{uv}(x_u, x_v),$$
where $\rho = \alpha + \sum_{j \notin \{u,v\}} f_j(0)$ and $\tilde{\rho} = \tilde{\alpha} + \sum_{j \notin \{u,v\}} \tilde{f}_j(0)$.  Using the result for $p=2$, this immediately implies that $f_u = \tilde{f}_u$ for all $u=1,\ldots,p$, $h_{uv} = \tilde{h}_{uv}$ for all $1\leq u < v \leq p$, and $\rho = \tilde{\rho}$, which implies $\alpha = \tilde{\alpha}$.

\end{proof}

\section{Other Approaches}
\label{supp:other}
We also considered alternative Bayesian monotone spline formulations based on restricting the sign of the coefficients \citep{ramsay1988monotone, shively2009bayesian} instead of squaring unconstrained functions. However, squaring unconstrained functions led to superior estimation performance in simulations, particularly when the non-negative function being estimated was close to $0$. To see why, suppose $g(x) = \{\sum_{j=1}^{m} s_j(x) c_j\}^2$, where $\mathbf{s}(x) = (s_j(x))_{j=1}^{m}$ denotes $m$ B-spline basis functions and $\mathbf{c} = (c_1, \ldots, c_m)^{\T}$ denotes the vector of basis coefficients. If we assume $\mathbf{c} \sim N(0, C)$ for some covariance matrix $C$, then for each $x$, $g(x) / \left[\mathbf{s}(x)^{\T} \, C \, \mathbf{s}(x)\right] \sim \chi_1^2$. Since the $\chi_1^2$ density function has an infinite spike at $0$, the squared model provides superior shrinkage of the non-negative function towards $0$ as compared to the model $g(x) = \mathbf{s}(x)^{\T} \mathbf{c}$, where we assume a truncated Gaussian prior  $\mathbf{c} \sim N(0, C) \mathbbm{1}[c_1 \geq 0, \ldots, c_m \geq 0]$. 

As an alternative to decomposing $P_{uv}$ and $N_{uv}$ as products of univariate non-negative functions, we tried modeling $P_{uv}$ and $N_{uv}$ as squares of linear combinations of tensor products of B-splines \citep{de1978practical}. That is, we let
\begin{align*}
    P_{uv}(x_u, x_v) & = \left\{\displaystyle\sum_{j=1}^{m} \sum_{j'=1}^{m} s_{uj}(x_u) s_{vj'}(x_v) \beta_{uv, jj'}\right\}^2, \text{ and }\\
    N_{uv}(x_u, x_v) & = \left\{\displaystyle\sum_{j=1}^{m} \sum_{j'=1}^{m} s_{uj}(x_u) s_{vj'}(x_v) \delta_{uv, jj'}\right\}^2.
\end{align*}
For $1 \leq u < v \leq p$, one may then analogously endow $\beta_{uv} = (\beta_{uv, jj'})_{jj'}$ and $\delta_{uv} = (\delta_{uv, jj'})_{jj'}$ with priors 
\begin{align*}
    \beta_{uv} \mid \tau_{uv,1}^2, \nu^2 & \sim N(0, \nu^2 \tau_{uv,1}^2 \tilde{\Sigma}_0),\\
    \delta_{uv} \mid \tau_{uv,2}^2, \nu^2 & \sim N(0, \nu^2 \tau_{uv,2}^2 \tilde{\Sigma}_0),\\
    \tau_{uv,1}, \tau_{uv,2}, \nu & \sim C^{+}(0,1),
\end{align*}
where $\tilde{\Sigma}_0$ is an appropriately chosen bivariate penalty matrix as outlined in \cite{lang2004bayesian}. We denote this approach by G-SAID and implement it for comparison with SAID and BKMR for the simulation case SINE in Section 4.2. The results when estimating $H({\bf x})$ using SAID, G-SAID, and BKMR are reported in Table \ref{tab:p_2_sin_whole}. Assuming $P_{uv}$ and $N_{uv}$ to be products of univariate non-negative functions is a special case of the above model, obtained by letting $\beta_{uv, jj'} = \theta_{uv, 1j} \phi_{uv, 1j'}$ and $\delta_{uv, jj'} = \theta_{uv, 2j} \phi_{uv, 2j'}.$ This is equivalent to assuming the matrices $\mathcal{B}_{uv} = (\beta_{uv, jj'})_{1 \leq j, j' \leq m}$ and $\mathcal{D}_{uv} = (\delta_{uv, jj'})_{1 \leq j, j' \leq m}$ are rank $1$. Due to the limited sample sizes and signal-to-noise ratios typical of environmental epidemiology, we found restricting the rank to $1$ to be much more effective  at identifying synergistic, antagonistic, or null interactions. Furthermore, the rank $1$ assumption reduces the number of spline coefficients for the $uv$th interaction from $2m^2$ to $4m$, leading to substantial computational gains.

We also tried including an additional penalty term in the prior distribution of the coefficients $\mathbf{\Psi}_{uv}$. 
For $1 \leq u < v \leq p$, we let $$\tilde{\mathcal{Q}}(P_{uv}, N_{uv}) = \int_{[0,1]^2} P_{uv}(x_1, x_2) dx_1 dx_2 \int_{[0,1]^2} N_{uv} (x_1, x_2) dx_1 dx_2,$$ and consider the following augmented prior on $\mathbf{\Psi}_{uv}$:
\begin{align*}
    \pi^*(\mathbf{\Psi}_{uv} \mid \tau_{uv,1}, \tau_{uv,2}, \kappa_{uv}, \nu) & \propto \pi(\mathbf{\Psi}_{uv} \mid \tau_{uv,1}, \tau_{uv,2}, \nu) \exp\left\{-\kappa_{uv} \tilde{\mathcal{Q}}(P_{uv}, N_{uv})\right\},
\end{align*}
where $\pi(\mathbf{\Psi}_{uv} \mid \tau_{uv,1}, \tau_{uv,2}, \nu)$ is as in Section 3.2 with the penalty parameter having prior $\log(\kappa_{uv}) \sim N(0,1)$ and $\tau_{uv,1}, \tau_{uv,2}, \nu \sim C^{+}(0,1)$ as before. We observe that $\tilde{\mathcal{Q}}(P_{uv}, N_{uv}) = 0$ if and only if either $P_{uv} \equiv 0$ or $N_{uv} \equiv 0$, implying $h_{uv}$ is either synergistic, antagonistic, or null (SAN). Thus, the augmented prior $\pi^*$ shrinks the interaction $h_{uv}$ towards the class of SAN interactions. However, based on numerous simulation experiments, we observed that the penalty led to little or no improvement when detecting $h_{uv}$, both from the perspective of estimation error and variable selection. We believe the already excellent performance of the half-Cauchy priors on the variance parameters in shrinking $(P_{uv}, N_{uv})_{1 \leq u < v \leq p}$ to zero to be the reason behind this. Furthermore, including the penalty leads to a doubly intractable posterior sampling problem requiring data augmentation rejection sampling \citep{rao2016data}, significantly complicating posterior computation. As a result, we did not proceed with the augmented prior approach.

\begin{table}
\centering
\begin{tabular}{|l|l|l|l|l|}
\hline
    Signal and Noise   & SAID  & BKMR & G-SAID \\ \hline
$ \gamma_0 = 1, \sigma_0^2 = 0.1$ & 0.17 & 0.12 &   0.10         \\ \hline
$\gamma_0 = 1, \sigma_0^2 = 0.5$ & 0.24 & 0.30 &   0.18        \\ \hline
$\gamma_0 = 2, \sigma_0^2 = 0.1$ & 0.32 & 0.12 &    0.12           \\ \hline
$\gamma_0 = 2, \sigma_0^2 = 0.5$ & 0.36 & 0.27 &  0.22              \\ \hline
\end{tabular}
\caption{RMSE of competitors when estimating H in case SINE. Here $n = 500$ and $p = 2$. Lower values indicate better performance.}
\label{tab:p_2_sin_whole}
\end{table}

\section{Hamiltonian Monte Carlo}
\label{supp:hmc}

We now provide further details on Hamiltonian Monte Carlo (HMC). For an extensive review, we refer the reader to \cite{neal2011mcmc, betancourt2017conceptual, betancourt2015hamiltonian}. Broadly, HMC uses gradient information to efficiently explore different parts of the distribution.

Suppose we are interested in sampling from a density $
\theta \sim G(\theta)$ supported on $\theta \in \Theta \subset \mathbb{R}^d$ with $d \geq 1$. For our purposes, $G(\cdot)$ is typically the posterior density of the parameter $\theta$. We introduce auxilliary variables $\mathbf{p} \sim N_d(0, \mathbf{M})$ independent of $\theta$, called {\em momentum} variables, and consider sampling from the joint distribution
$$G_1(\theta, \mathbf{p}) = G(\theta) N_d(\mathbf{p} \mid 0, \mathbf{M}).$$
The matrix $\mathbf{M}$ is called the {\em mass matrix}. We next consider the Hamiltonian, defined as
$$H(\theta, \mathbf{p}) = -\log G_1(\theta, \mathbf{p}) = -\log G(\theta) + \dfrac{1}{2} \mathbf{p}^{\T} \mathbf{M}^{-1} \mathbf{p} + \dfrac{d}{2} \log(2\pi) + \dfrac{1}{2} \log |\mathbf{M}|.$$
In a closed system with a particle having position $\theta$ and momentum $p$, the Hamiltonian is analogous to the total energy of the system. HMC considers simulating the trajectory of a particle with position $\theta = \theta(t)$ and momentum $\mathbf{p} = \mathbf{p}(t)$ as a function of evolving time $t$, such that the total energy or Hamiltonian $H(\theta, \mathbf{p}) = H(\theta(t), \mathbf{p}(t))$ remains constant. The trajectory of such a particle is dictated by Hamilton's equations, given by
\begin{align*}
    \dfrac{d\theta}{dt} & = \dfrac{d\theta(t)}{dt} = \dfrac{\partial H(\theta, \mathbf{p})}{\partial \mathbf{p}} = \mathbf{M}^{-1} \mathbf{p},\\
    \dfrac{d\mathbf{p}}{dt} & = \dfrac{d\mathbf{p}(t)}{dt} = -\dfrac{\partial H(\theta, \mathbf{p})}{\partial \theta} = \dfrac{d}{d\theta} \log G(\theta).
\end{align*}
One could then simulate $ (\theta, \mathbf{p}) = (\theta(t), \mathbf{p}(t))$ by solving Hamilton's equations. Unfortunately, these equations do not usually admit analytical solutions and thus require numerical approximations. One such approach is the Stormer-Verlet integrator \citep{duane1987hybrid}, commonly known as the ``leapfrog'' operator. At a particular iteration of the MCMC algorithm, suppose the position is given by $\theta^{(0)} = \theta{(0)}$. One first draws $\mathbf{p}^{(0)} = \mathbf{p}{(0)} \sim N_d(0, \mathbf{M})$. The leapfrog operator then iterates 
\begin{align*}
    \mathbf{p}{(t + 0.5)} & = \mathbf{p}{(t)} + \dfrac{e_0}{2} \dfrac{d}{d\theta} \log G(\theta) \bigg|_{\theta = \theta{(t)}},\\
    \theta{(t+1)} & = \theta{(t)} + e_0 \mathbf{M}^{-1} \mathbf{p}{(t + 0.5)},\\
    \mathbf{p}{(t + 1)} & = \mathbf{p}{(t + 0.5)} + \dfrac{e_0}{2} \dfrac{d}{d\theta} \log G(\theta) \bigg|_{\theta = \theta{(t + 1)}},
\end{align*}
for $t = 0, 1, \ldots, L_0-1$. Here, $e_0$ is called the step-size and $L_0$ the number of steps. Let the output from this algorithm be $(\theta^*, \mathbf{p}^*)$. As the leapfrog operator is a numerical approximation, the obtained $(\theta^*, \mathbf{p}^*)$ might not preserve the Hamiltonian. To ensure the resulting Markov chain is ergodic and time reversible, we use a Metropolis-Hastings \citep{metropolis1953equation, hastings1970monte} style accept-reject step that accepts $(\theta^*, \mathbf{p}^*)$ with probability
$$\alpha = \min \left[1, \exp\left\{-H(\theta^*, \mathbf{p}^*) + H\left(\theta^{(0)}, \mathbf{p}^{(0)}\right)\right\}\right].$$
This approach is now iterated to yield the desired number of MCMC samples for $(\theta, \mathbf{p})$. One then simply discards the samples of $\mathbf{p}$ to obtain MCMC samples of $\theta$.

We remark here that for carrying out HMC to sample from $G(\cdot)$, one does not require the exact form of $G(\cdot)$. Instead, a proportionality $G(\theta) \propto g(\theta)$ suffices, where $G(\theta) = g(\theta) / \int_{\Theta} g(\theta) d\theta$, as the unknown normalizing constant $\int_{\Theta} g(\theta) d\theta$ gets eliminated in Hamilton's equations. This is particularly helpful when sampling from the posterior density $\Pi(\theta \mid \mathcal{Y}) \propto \pi(\theta) p(\mathcal{Y} \mid \theta)$, as typically the normalizing constant $p(\mathcal{Y}) = \int_{\Theta} p(\mathcal{Y} \mid \theta) \pi(\theta) d\theta$ is not analytically tractable.

Performance of HMC crucially depends on carefully choosing the tuning hyperparameters $L_0, e_0,$ and $\mathbf{M}$.  In fitting SAID, we estimate $\mathbf{M}$ as the empirical covariance of the gradient $d \log G(\theta) / d\theta$ using a pilot run, and observed $L_0 \sim 10$ and $e_0 \approx 0.01$ to work well. We also perturb $e_0$ every $500-1000$ MCMC iterations to ensure the algorithm does not get stuck.  

\section{Further tables on simulation results for $p=2$}
\label{supp:p_2_table}

We provide further tables for the results described in Section 4.2. Tables \ref{tab:p_2_case2_whole} and \ref{tab:p_2_case2_int} provide RMSE values for estimating the dose response surface $H$ and the interaction $h_{12}$, respectively, under the QR setup. Tables \ref{tab:p_2_case3_whole} and \ref{tab:p_2_case3_int} provide RMSE values for estimating the dose response surface $H$ and the interaction $h_{12}$, respectively, under the MIS setup. Table \ref{tab:p_2_case3_varsel} provides the proportion of times each method misclassifies the detected interaction as null in the MIS scenario, obtained over $R = 100$ replicates. For the Bayesian approaches SAID and MixSelect, we declare the interaction to be non-null if the posterior inclusion probability (PIP) is greater than or equal to $0.5$. Details on the QR and MIS setups are provided in Section 4.2.

\begin{table}[H]
\centering
\begin{tabular}{|l|l|l|l|l|l|l|l|}
\hline
    Signal and Noise   & SAID  & BKMR & MixSelect & HierNet & FAMILY & PIE  & RAMP \\ \hline
$ \gamma_0 = 1, \sigma_0^2 = 0.1$ & 0.05 & 0.04 &  0.03         & 0.17   & 0.43   & 0.20 &   0.03   \\ \hline
$\gamma_0 = 1, \sigma_0^2 = 0.5$ & 0.11 & 0.08 &   0.07        & 0.18    & 0.43   & 0.25 &  0.08 \\ \hline
$\gamma_0 = 2, \sigma_0^2 = 0.1$ & 0.05 & 0.04 & 0.04           & 0.18    & 0.57   & 0.60 & 0.03     \\ \hline
$\gamma_0 = 2, \sigma_0^2 = 0.5$ & 0.11 & 0.08 & 0.07          & 0.19    & 0.57   & 0.59 & 0.07  \\ \hline
\end{tabular}
\caption{RMSE of competing methods in estimating $H$ for scenario QR. Here $n = 500$ and $p = 2$. Lower values indicate better performance.}
\label{tab:p_2_case2_whole}
\end{table}
\begin{table}[H]
\centering
\begin{tabular}{|l|l|l|l|l|l|l|l|}
\hline
  Signal and Noise    & SAID   & MixSelect & HierNet & FAMILY & PIE  & RAMP \\ \hline
$\gamma_0 = 1, \sigma_0^2 = 0.1$ & 0.10  &  0.10         &  0.33   & 0.31  & 0.04  &    0.04  \\ \hline
$\gamma_0 = 1, \sigma_0^2 = 0.5$ & 0.20  &   0.15        & 0.31    & 0.31   & 0.11 &  0.20    \\ \hline
$\gamma_0 = 2, \sigma_0^2 = 0.1$ & 0.13  &    0.10       & 0.41    & 0.62   & 0.04 & 0.04     \\ \hline
$\gamma_0 = 2, \sigma_0^2 = 0.5$ & 0.20  &  0.15         & 0.42    & 0.62   & 0.10 & 0.10  \\ \hline
\end{tabular}
\caption{RMSE of competing methods in estimating $h_{12}$ for scenario QR. Here $n = 500$ and $p = 2$. Lower values indicate better performance.}
\label{tab:p_2_case2_int}
\end{table}


\begin{table}[H]
\centering
\begin{tabular}{|l|l|l|l|l|l|l|l|}
\hline
    Signal and Noise   & SAID  & BKMR & MixSelect & HierNet & PIE  & RAMP \\ \hline
$ \gamma_0 = 1, \sigma_0^2 = 0.1$ & 0.05 & 0.05 &  0.10         & 0.19   &  1.43 &  0.08   \\ \hline
$\gamma_0 = 1, \sigma_0^2 = 0.5$ & 0.11 & 0.11 &    0.12       & 0.19       & 1.18 &  0.12 \\ \hline
$\gamma_0 = 2, \sigma_0^2 = 0.1$ & 0.05 & 0.05 & 0.17           & 0.26       & 2.48 & 0.13     \\ \hline
$\gamma_0 = 2, \sigma_0^2 = 0.5$ & 0.12 & 0.12 & 0.19          & 0.26       & 2.32 & 0.16  \\ \hline
\end{tabular}
\caption{RMSE of competing methods in estimating $H$ for scenario MIS. Here $n = 500$ and $p = 2$. Lower values indicate better performance.}
\label{tab:p_2_case3_whole}
\end{table}
\begin{table}[H]
\centering
\begin{tabular}{|l|l|l|l|l|l|l|l|}
\hline
  Signal and Noise    & SAID   & MixSelect & HierNet  & PIE  & RAMP \\ \hline
$\gamma_0 = 1, \sigma_0^2 = 0.1$ & 0.11  &   0.38        &  0.15   & 0.33    &    0.33  \\ \hline
$\gamma_0 = 1, \sigma_0^2 = 0.5$ & 0.13  &   0.29        & 0.15    & 0.32    &  0.30    \\ \hline
$\gamma_0 = 2, \sigma_0^2 = 0.1$ & 0.20  &  0.67         & 0.33    & 0.65    & 0.65     \\ \hline
$\gamma_0 = 2, \sigma_0^2 = 0.5$ & 0.21  &   0.55        & 0.33    & 0.65    & 0.65  \\ \hline
\end{tabular}
\caption{RMSE of competing methods in estimating $h_{12}$ for scenario MIS. Here $n = 500$ and $p = 2$. Lower values indicate better performance.}
\label{tab:p_2_case3_int}
\end{table}

\begin{table}[H]
\centering
\begin{tabular}{|l|l|l|l|l|l|l|l|}
\hline
  Signal and Noise  (SNR)  & SAID   & MixSelect & HierNet  & PIE  & RAMP \\ \hline
$\gamma_0 = 1, \sigma_0^2 = 0.1$ $(0.15)$ & 0.02  &   0.00        &  0.10   & 0.00    &    0.00  \\ \hline
$\gamma_0 = 1, \sigma_0^2 = 0.5$ $(0.03)$ & 0.27  &   0.16        & 0.69    & 0.20    &  0.37    \\ \hline
$\gamma_0 = 2, \sigma_0^2 = 0.1$ $(0.62)$ & 0.00  &  0.03         & 0.00    & 0.00    & 0.00     \\ \hline
$\gamma_0 = 2, \sigma_0^2 = 0.5$ $(0.12)$ & 0.00  &   0.01        & 0.02    & 0.00    & 0.00  \\ \hline
\end{tabular}
\caption{Proportion of misclassifying $h_{12}$ as null for each competing method under scenario MIS. Here $n = 500$ and $p = 2$. Signal-to-noise ratio (SNR) for detecting $h_{12}$ is also indicated. Lower values indicate better performance.}
\label{tab:p_2_case3_varsel}
\end{table}

\section{Creatinine Data Application}
\label{supp:application}

We provide plots for rest of the main effects and interaction effects not presented in Section 5.3. In Figure \ref{suppfig:me-plots}, we plot the estimated main effects of the metals Antimony, Barium, Cobalt, Lead, Manganese, Strontium, Thallium, Tin, and Tungsten. In Figure \ref{suppfig:ie_plots}, we plot the estimated interaction surface between Cadmium and Manganese. 

\begin{figure}
\centering
\begin{tabular}{ccc}
    \includegraphics[width=10em]{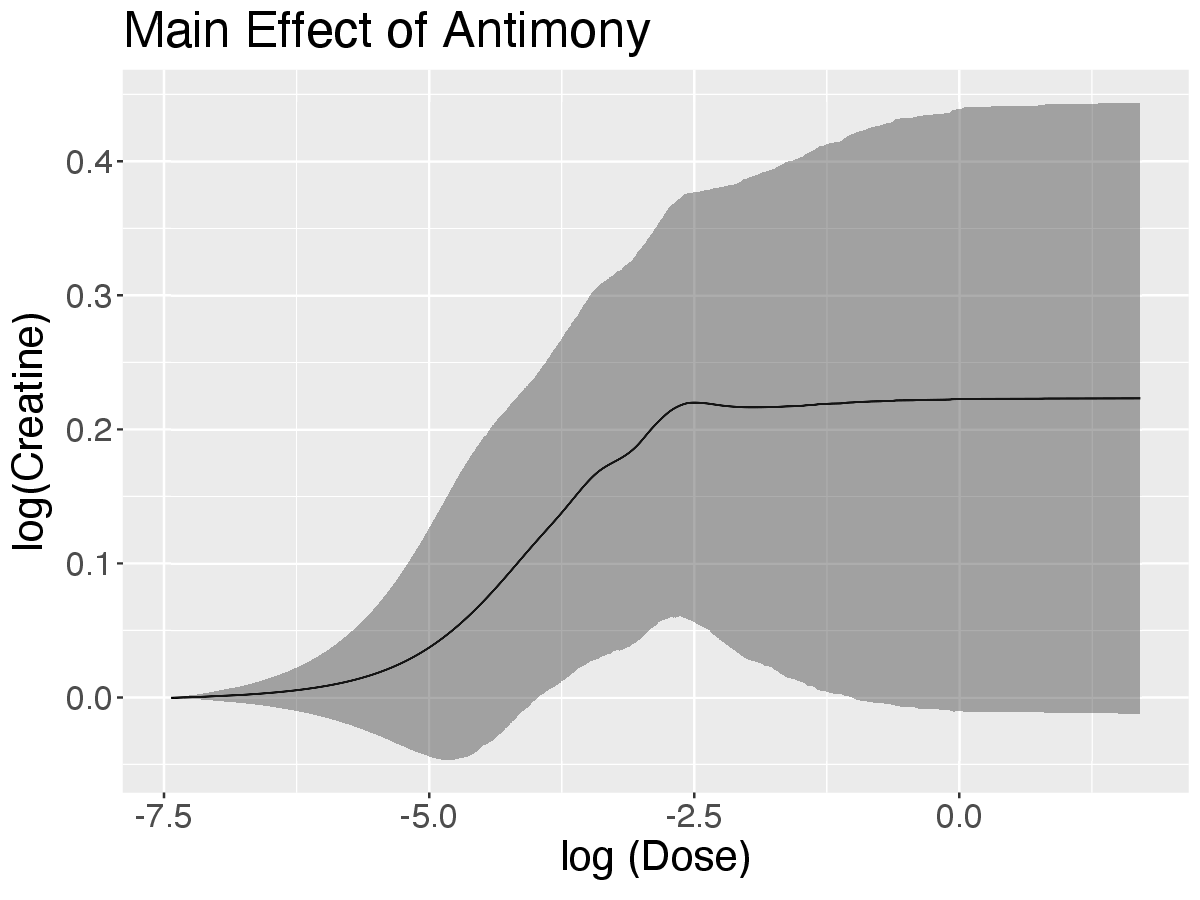} &
     \includegraphics[width=10em]{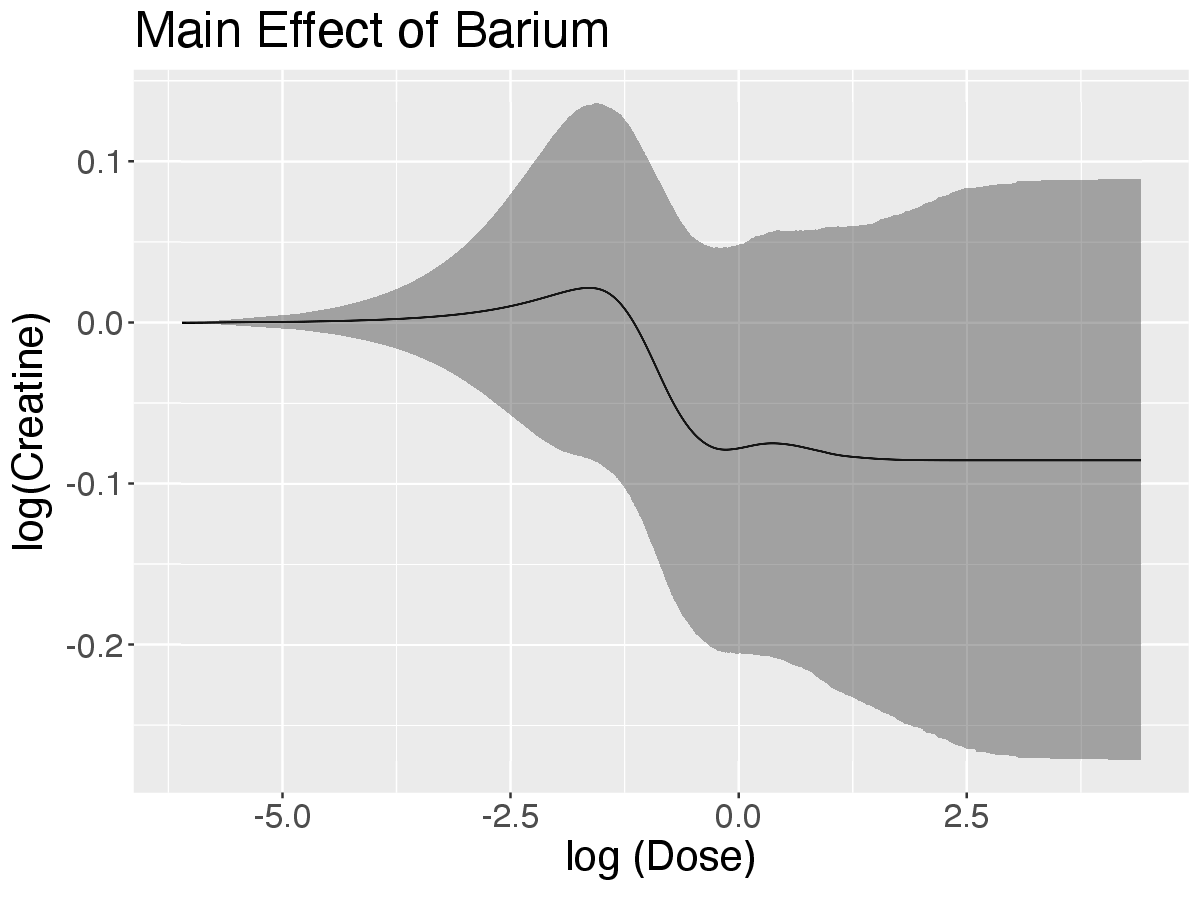} \\
      \includegraphics[width=10em]{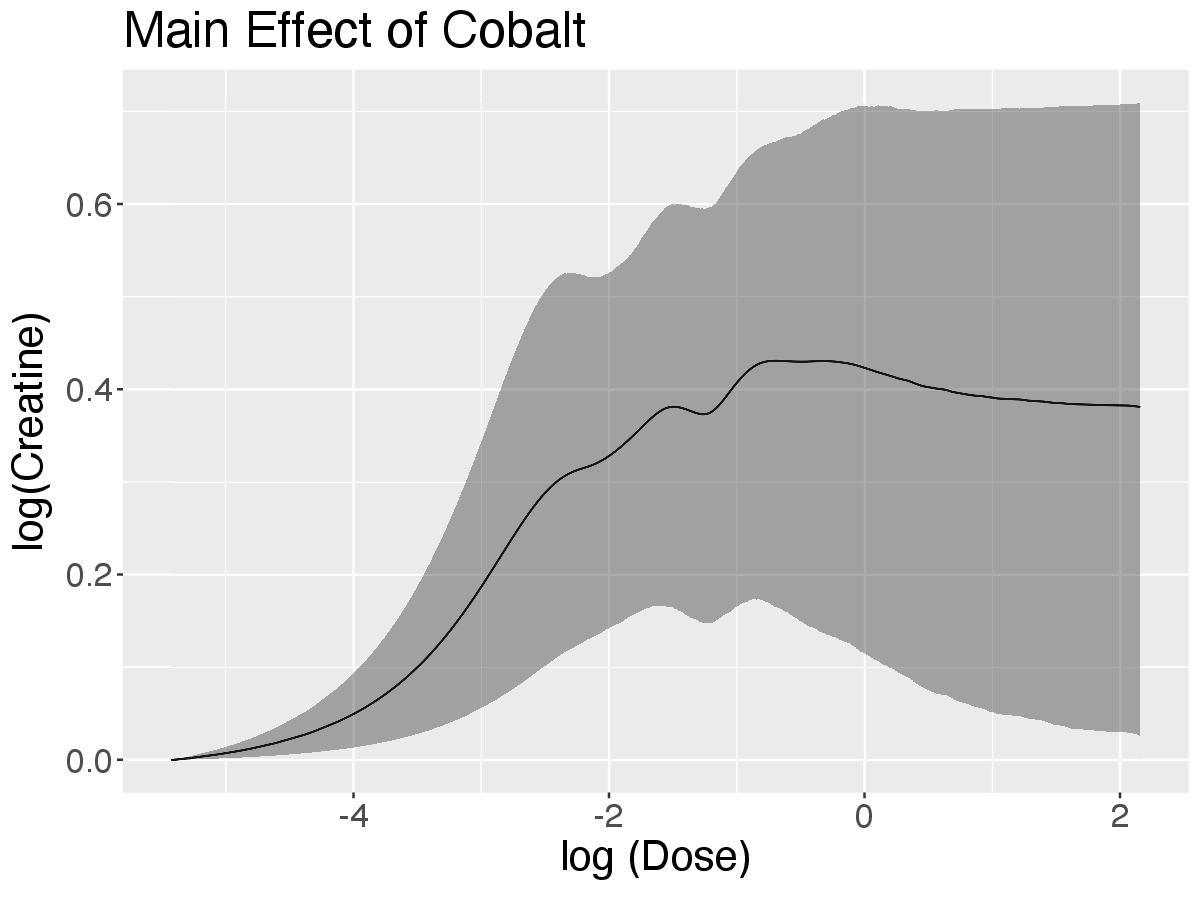}  &  \includegraphics[width=10em]{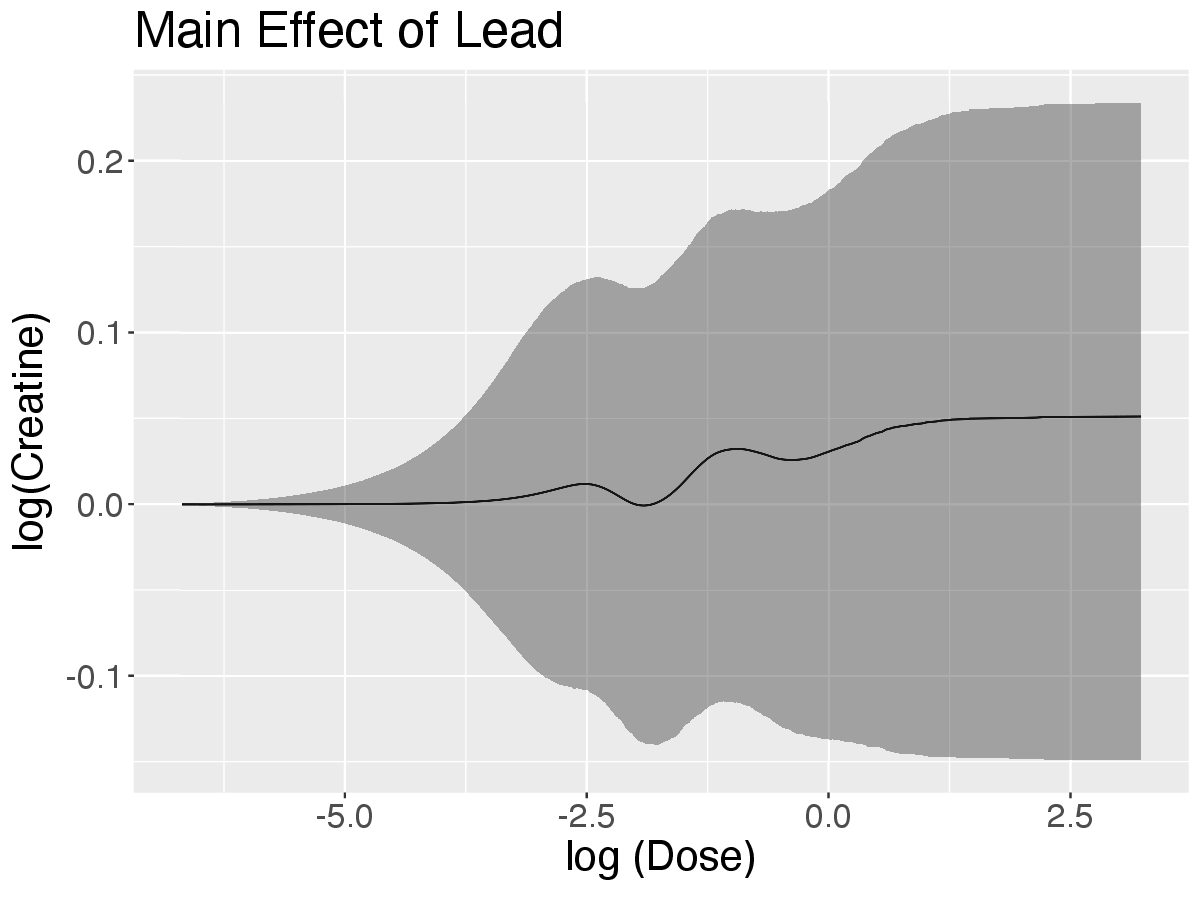} \\
       \includegraphics[width=10em]{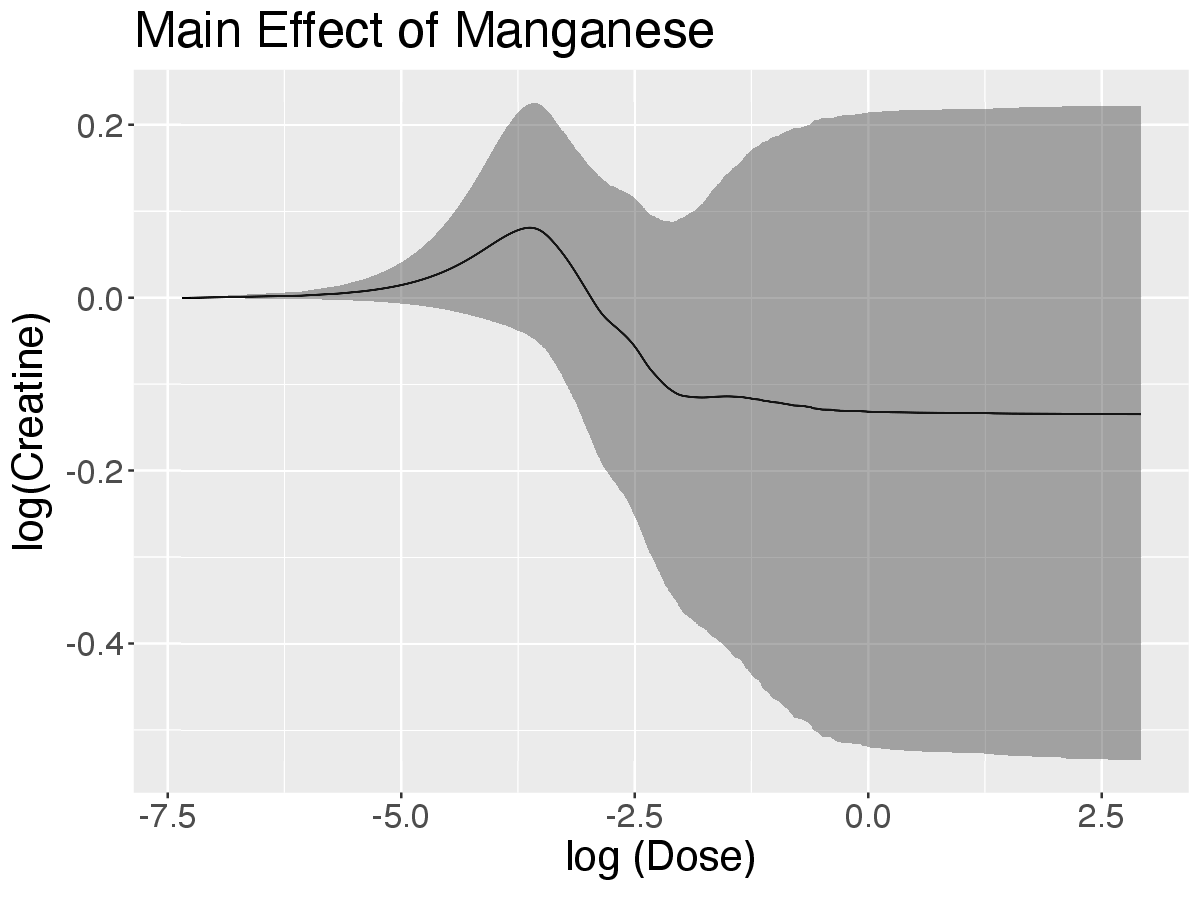} &   \includegraphics[width=10em]{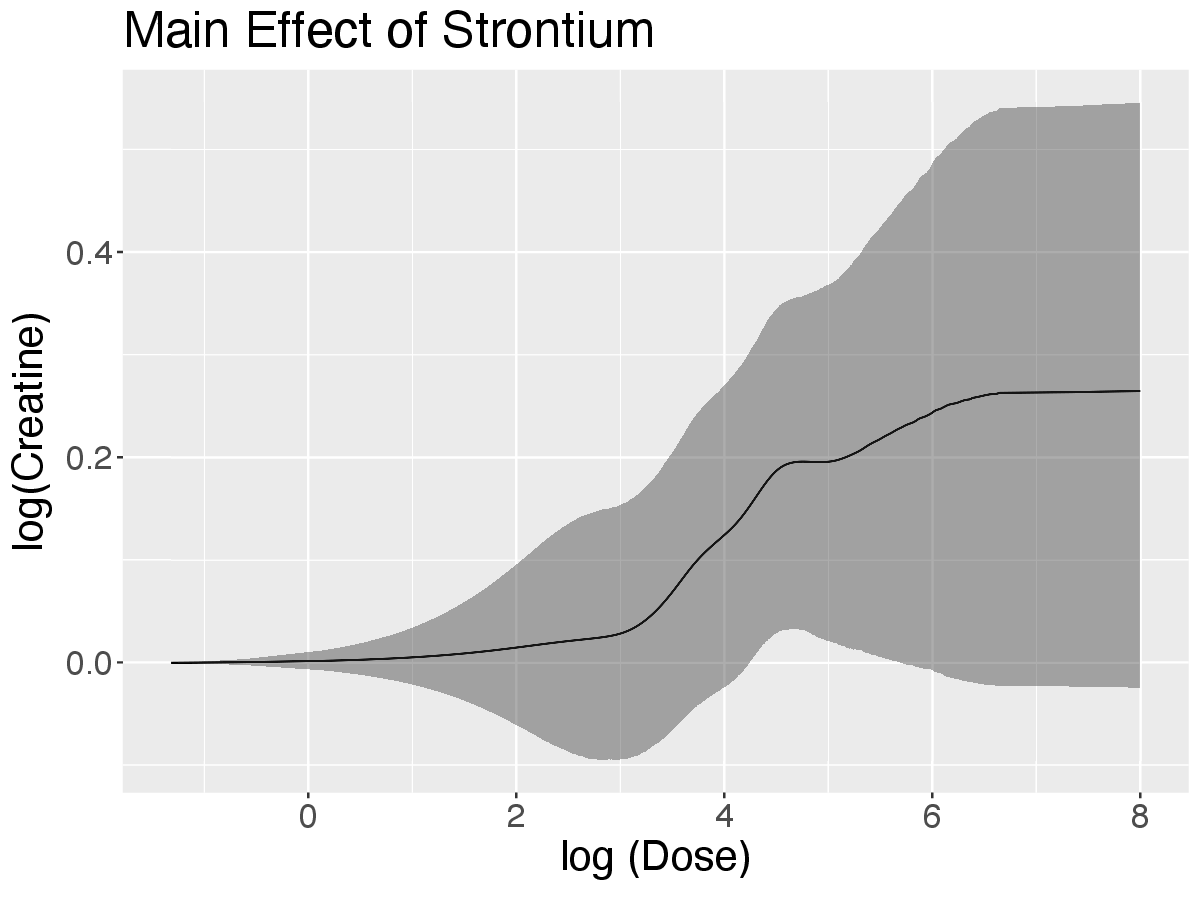} \\
        \includegraphics[width=10em]{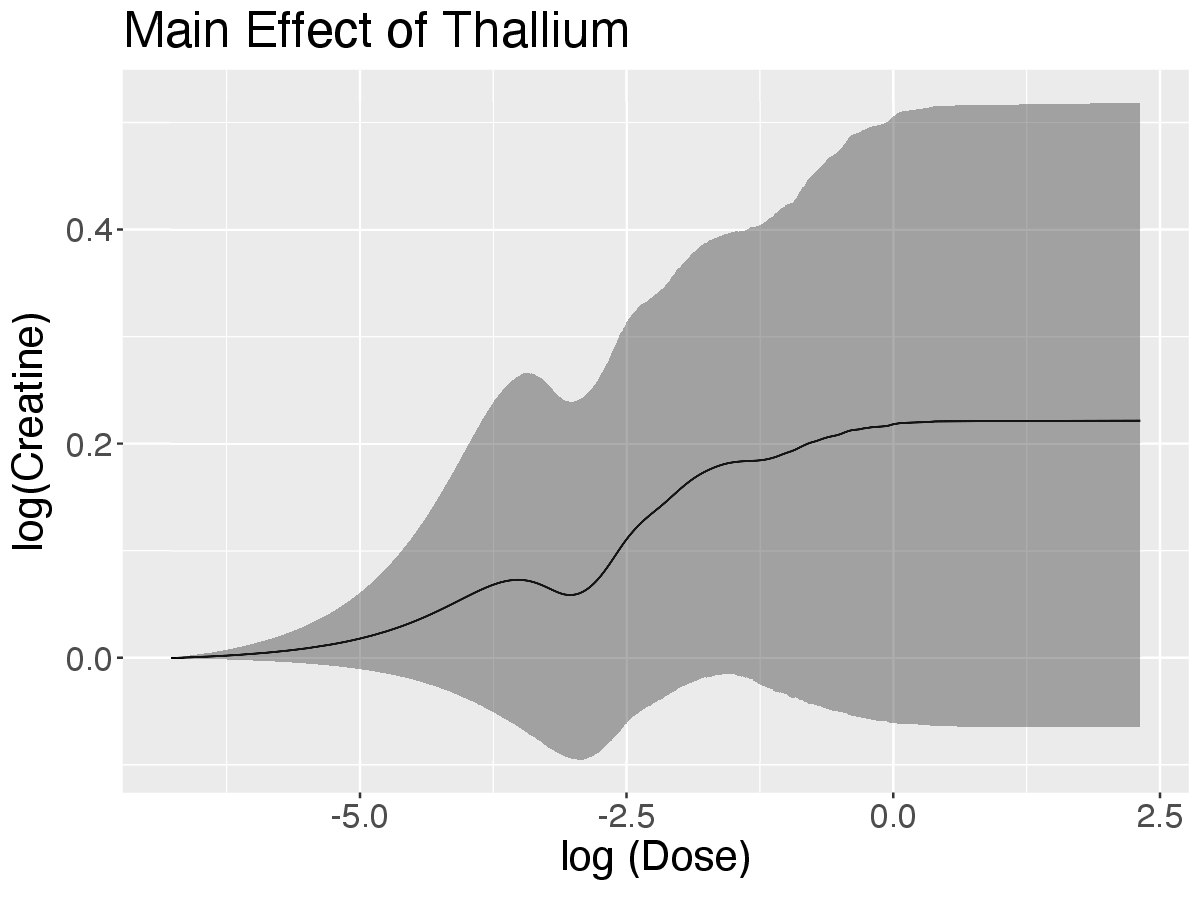} &   \includegraphics[width=10em]{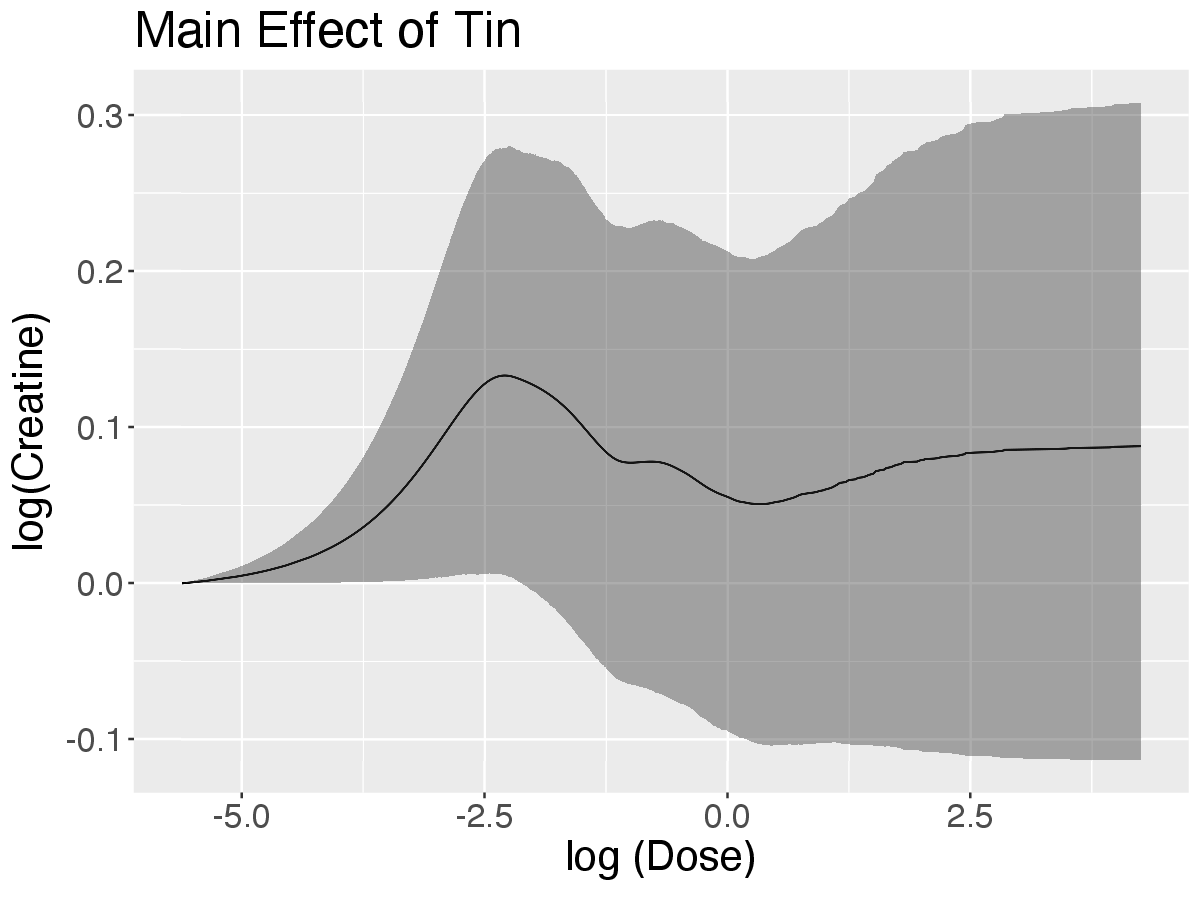} \\
         \includegraphics[width=10em]{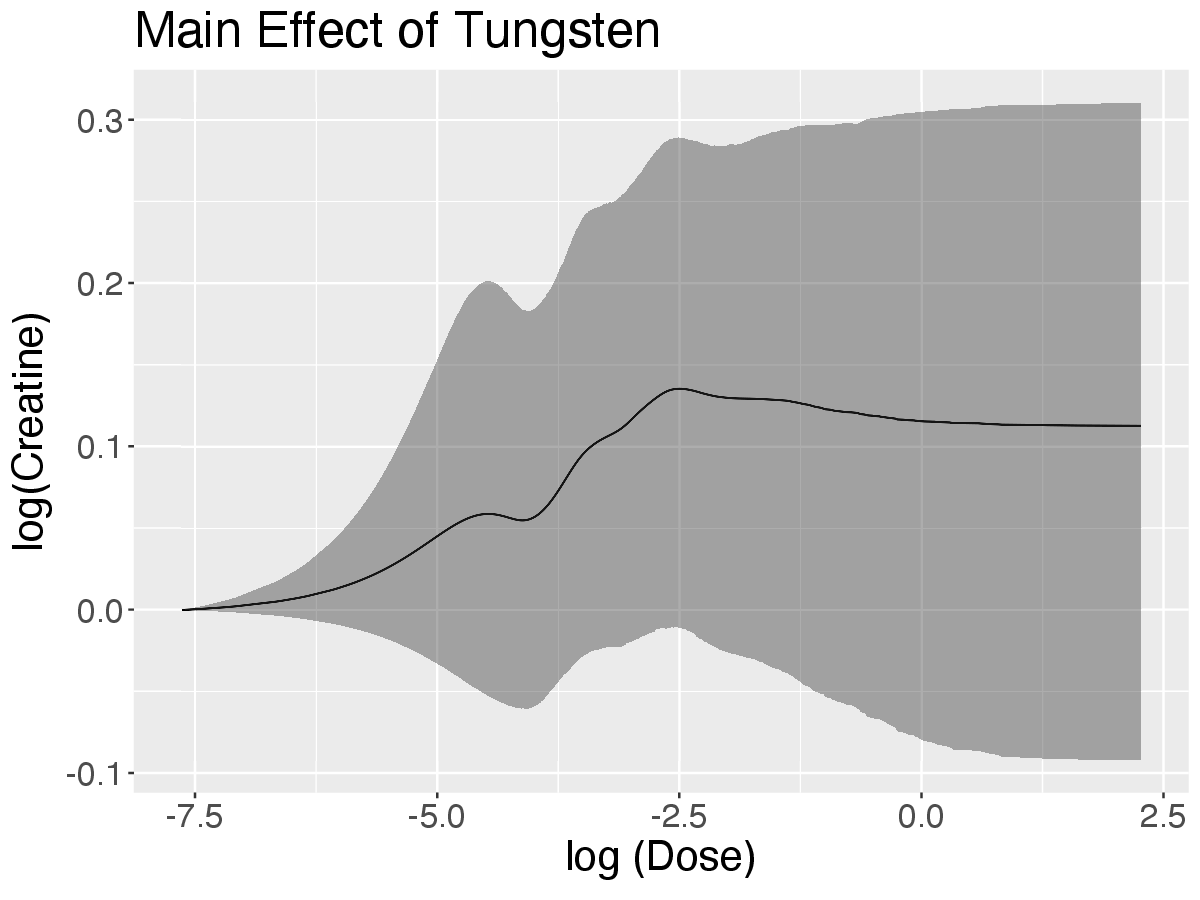} 
\end{tabular}
\caption{Plots showing the main effects of dilution-adjusted Antimony, Barium, Cobalt, Lead, Manganese, Strontium, Thallium, Tin, and Tungsten on log dilution-adjusted Creatinine. Exposure levels are in log scale. Black line denotes posterior mean and shaded regions denote pointwise $95 \%$ posterior credible intervals. }
\label{suppfig:me-plots}
\end{figure}




\begin{figure}
    \centering
    \includegraphics[width=35em]{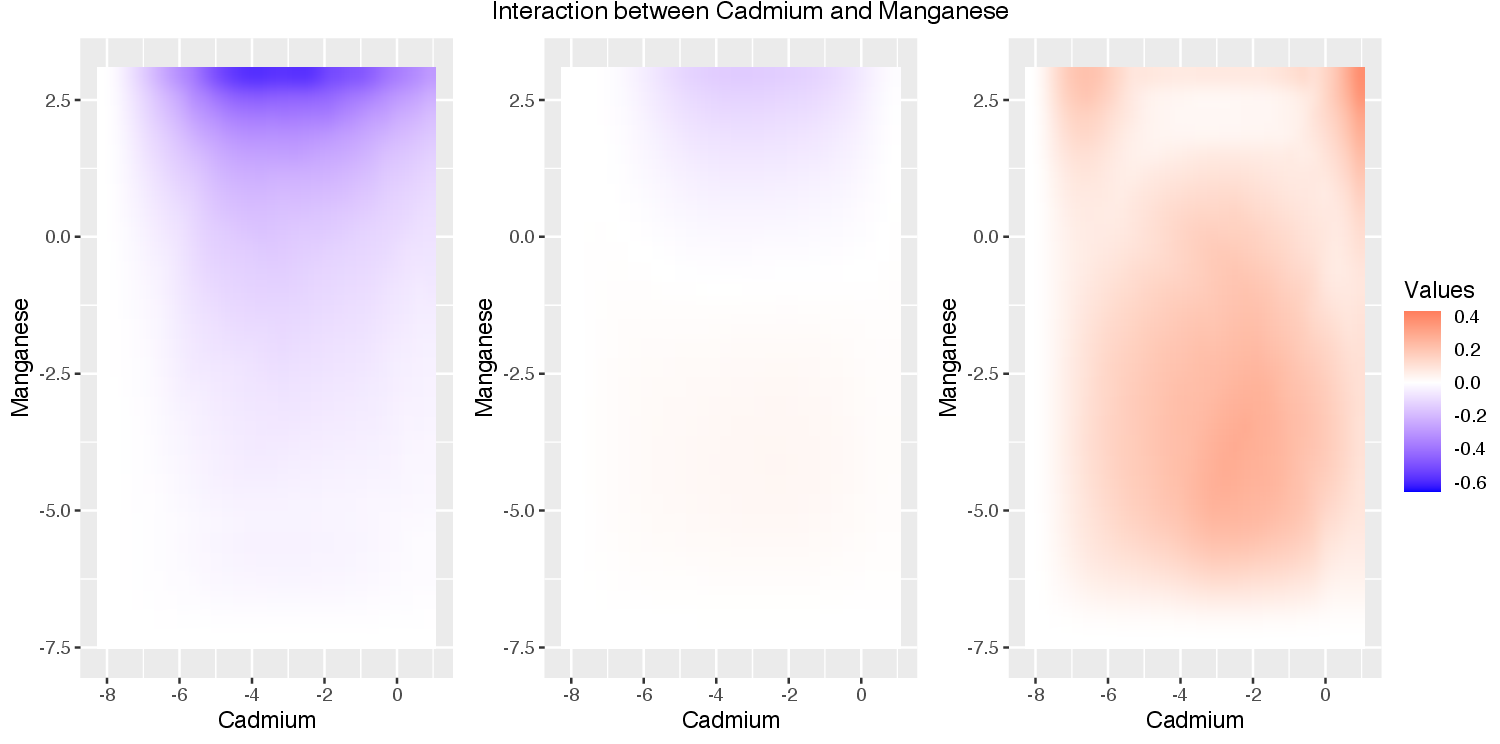}
    \caption{Plot showing the interaction effect of dilution-adjusted Cadmium and Manganese on log dilution-adjusted  Creatinine. Exposure levels are in log scale. Plot shows the pointwise $2.5 \%$ posterior credible surface, the posterior mean, and the $97.5 \%$ posterior credible surface from left to right.}
    \label{suppfig:ie_plots}
\end{figure}





\end{document}